\documentclass[a4paper,12pt]{article}
\pdfoutput=1 % if your are submitting a pdflatex (i.e. if you have
\usepackage{jheppub} % for details on the use of the package, please
\usepackage[T1]{fontenc} % if needed

\usepackage{float, array, xspace, amscd, amsthm}
\usepackage{fancyhdr}
\usepackage{longtable}

\newtheorem{theorem}{Theorem}
\newtheorem{lemma}{Lemma}
\newtheorem{proposition}{Proposition}

\newtheorem{definition}{Definition}

\newcommand{\nn}{\nonumber}
\newcommand{\dd}{{\rm d}}
\newcommand{\w}{\wedge}
\newcommand{\p}{\partial}
\newcommand{\End}{{\rm End}}
\newcommand{\cd}{\check \dd}

\newcommand{\tr}{{\rm tr}}
\newcommand{\dth}{{{\dd}}_\theta}
\newcommand{\cdth}{{{\check\dd}}_\theta}

\newcommand{\cN}{\mathcal{N}}

\newcommand{\be}{\begin{equation}}
\newcommand{\ee}{\end{equation}}
\def\bea#1\eea{\begin{align}#1\end{align}}

\usepackage{color}

\title{
Infinitesimal moduli of G2 holonomy manifolds with instanton bundles
}

\author[a]{Xenia de la Ossa,}
\author[b]{Magdalena Larfors,}
\author[c,d,e]{Eirik E.~Svanes}

\affiliation[a]{Mathematical Institute, Oxford University\\Andrew Wiles Building, Woodstock Road\\Oxford OX2 6GG, UK }
\affiliation[b]{Department of Physics and Astronomy,Uppsala University\\ SE-751 20 Uppsala, Sweden}
\affiliation[c]{Sorbonne Universit\'es, UPMC Univ. Paris 06, UMR 7589, LPTHE, F-75005, Paris, France}
\affiliation[d]{CNRS, UMR 7589, LPTHE, F-75005, Paris, France}
\affiliation[e]{Sorbonne Universit\'es, Institut Lagrange de Paris, 98 bis Bd Arago, 75014 Paris, France\\}

\emailAdd{delaossa@maths.ox.ac.uk, magdalena.larfors@physics.uu.se, esvanes@lpthe.jussieu.fr}

\abstract{We describe the infinitesimal moduli space of pairs $(Y, V)$ where $Y$ is a manifold with $G_2$ holonomy, and $V$ is a vector bundle on $Y$ with an instanton connection. These structures arise in connection to the moduli space of heterotic string compactifications on compact and non-compact seven dimensional spaces, e.g. domain walls. Employing the canonical $G_2$  cohomology developed by Reyes-Carri\'on and Fern\'andez and Ugarte, we show that the moduli space decomposes into the sum of the bundle moduli $H^1_{\cd_A}(Y,\End(V))$ plus the moduli of the $G_2$ structure preserving the instanton condition. The latter piece is contained in $H^1_{\cd_\theta}(Y,TY)$, and is given by the kernel of a map ${\cal\check F}$ which generalises the concept of the Atiyah map for  holomorphic bundles on complex manifolds to the case at hand.  In fact, the map ${\cal\check F}$ is given in terms of the curvature of the bundle and maps $H^1_{\cd_\theta}(Y,TY)$  into $H^2_{\cd_A}(Y,\End(V))$, and moreover can be used to define a cohomology on an extension bundle of $TY$ by $\End(V)$. We comment further on the resemblance with the holomorphic Atiyah algebroid and connect the story to physics, in particular to heterotic compactifications on $(Y,V)$ when $\alpha'=0$.}

\begin{document}

\maketitle
\flushbottom

\newpage

% ---------------------------------------- Main text

\section{Introduction}

Manifolds with special holonomy have, since long, been used to construct supersymmetric lower-dimensional vacuum solutions of string and M theory. Seven-dimensional manifolds with $G_2$ holonomy are of interest for two types of vacua: Firstly, compact  $G_2$ holonomy manifolds may be used as the internal space in M theory constructions of four-dimensional vacua preserving $\cN=1$ supersymmetry. Secondly, non-compact $G_2$ holonomy manifolds have been used to construct four-dimensional $\cN=1/2$ BPS domain wall solutions of the heterotic string. In both types of configurations, the moduli space of the compactification is of fundamental importance for the lower dimensional model.

In the mathematical literature, $G_2$ manifolds were first discussed by Berger \cite{berger55}, and the first examples of $G_2$ metrics were constructed by Bryant \cite{MR916718}, Bryant--Salamon \cite{bryant89} and Joyce \cite{joyce1996:1,joyce1996:2}. 
Deformations of  $G_2$ holonomy manifolds, and their associated moduli space, have been thoroughly studied, both by mathematicians and theoretical physicists \cite{joyce1996:1,joyce1996:2,2003math......1218K,Dai2003,deBoer:2005pt,2007arXiv0709.2987K,Grana:2014vxa}
(see \cite{Grigorian:2009ge} for a recent review). It has been shown, by Joyce \cite{joyce1996:1,joyce1996:2}, that, for compact spaces, the third Betti number sets the dimension of the infinitesimal moduli  space.\footnote{See \cite{2012arXiv1212.6457K}  for a recent discussion of deformations of non-compact $G_2$ holonomy manifolds. The study of large deformations of $G_2$ holonomy manifolds is complicated by the fact that the deformation may lead to a torsionful $G_2$ structure \cite{2003math......1218K}. In this paper, we restrict to infinitesimal deformations of $G_2$ holonomy manifolds, and will return to the topic of deformations of torsionful $G_2$ structures in a companion paper \cite{delaOssa}.} This space may be endowed by a metric \cite{Hitchin:2000jd,Gutowski:2001fm,Beasley:2002db},  that shares certain properties with the K\"ahler metric on a Calabi--Yau moduli space \cite{deBoer:2005pt,2007arXiv0709.2987K}. In particular, when used in M theory compactifications, Grigorian and Yau \cite{Grigorian:2008tc} have proposed a local K\"ahler metric for the combined deformation space of the geometry and  M theory flux potential.

However, to the best of our knowledge, the moduli space of the  $G_2$ structure manifolds needed for heterotic BPS domain walls of  \cite{deCarlos:2005kh,Gurrieri:2004dt,Gurrieri:2007jg,Klaput:2011mz,Gray:2012md,Klaput:2012vv,Klaput:2013nla,Lukas:2010mf,Gemmer:2013ica} remains largely to be explored.\footnote{See \cite{Gran:2007kh,Papadopoulos:2008rx,Papadopoulos:2009br,Gutowski:2014ova,Beck:2015gqa,Gran:2016zxk} for discussions on the classification of this type of heterotic and M theory vacua.} In this paper, we will focus on this topic. Our study follows up on our recent paper \cite{delaOssa:2014lma}, where the moduli space of certain six-dimensional $SU(3)$ structure manifolds was explored using an embedding manifold with $G_2$ structure. Here, we take a different perspective and study the moduli space of $G_2$ holonomy manifolds together with that of a vector bundle that encodes the heterotic gauge field. As we will discuss in section \ref{sec:instbund}, supersymmetry translates into an instanton condition on the vector bundle. Deformations of instanton bundles over $G_2$ manifolds have been studied before, see e.g. \cite{reyes1993some, carrion1998generalization, donaldson1998gauge, 2009arXiv0902.3239D, sa2009instantons}, and deformation studies of $G$ structures with instantons also appeared recently in \cite{2014SIGMA..10..083S, Bunk:2014ioa, Charbonneau:2015coa}.

In this article, we will construct the infinitesimal moduli space of  the system $(Y, V)$, where $Y$ is a manifold with $G_2$ holonomy and $V$ is a vector bundle on $Y$ with an instanton connection. This is a well-defined mathematical problem, and provides a first approximation of the geometry and bundle relevant for heterotic $\cN=1/2$ BPS solutions. Our main result is that the infinitesimal moduli space of this system is restricted to lie in the kernel of a map $\check{\mathcal F}$ in the canonical $G_2$ cohomology of \cite{carrion1998generalization,reyes1993some,fernandez1998dolbeault}. We thus show that the so-called Atiyah map stabilisation mechanism for Calabi--Yau moduli in $\cN=1$ heterotic string vacua, which was first discussed by Anderson {\it et.al} \cite{Anderson:2010mh,Anderson:2011ty, Anderson:2016byt}, may be extended to less supersymmetric configurations. We term this map the {\it $G_2$ Atiyah map}, in analogy with the corresponding map in Dolbeault cohomology on complex manifolds with holomorphic vector bundles. 

Recently, a sequence of papers \cite{delaOssa:2014cia,delaOssa:2014msa,Anderson:2014xha,Garcia-Fernandez:2015hja}
, two of which written by two of the present authors, have shed new light on the Atiyah stabilisation mechanism in $\cN=1$ heterotic string vacua. Due to the heterotic anomaly condition, which relates the gauge field strength, tangent bundle curvature to the $H$-flux of the Kalb--Ramond $B$-field, the infinitesimal moduli space is restricted to a more intricate nested kernel in Dolbeault cohomology, which is most conveniently encoded as a holomorphic structure on an extension bundle. This $\cN=1$ result is also of importance for the development of a generalised geometry for the heterotic string
\cite{Hohm:2011ex,Garcia-Fernandez:2013gja,Baraglia:2013wua,Bedoya:2014pma,Coimbra:2014qaa,Hohm:2014eba,Hohm:2014xsa,Garcia-Fernandez:2015hja}.
We expect to obtain similar result for the $\cN=1/2$ compactifications, once we allow $H$ flux. We will return to a study of this system, which corresponds to instanton bundles on manifolds with so-called integrable $G_2$ structure, in the companion paper \cite{delaOssa}. Let us remark already now that, to a large degree, the new results of this paper carry through to this general case.

We also mention that when finalising the current paper, an article appeared on ArXiv \cite{Clarke:2016qtg}, wherein the authors compute the infinitesimal moduli space of seven-dimensional heterotic compactifications and show by means of elliptic operator theory that the resulting space is finite dimensional. They also relate the resulting geometric structures to generalised geometry in a similar fashion to the six-dimensional Strominger system \cite{Garcia-Fernandez:2015hja}. Our approach to the problem resembles more that of \cite{delaOssa:2014cia,delaOssa:2014msa,Anderson:2014xha}, and it would be very interesting to compare with the findings of \cite{Clarke:2016qtg}, as can be done in the six-dimensional case. 

The structure of this paper is as follows. In section \ref{sec:g2struc} we recall the basic properties of manifolds with $G_2$ structure, and review the cohomologies that may be defined on such spaces. In particular, we introduce the canonical $G_2$ cohomologies $H_{\check \dd}^*(Y)$ and $H_{\check \dd_{\theta}}^*(Y,TY)$ for differential forms with values in the reals and the tangent bundle $TY$, respectively. Section \ref{sec:instbund} contains a review of instanton bundles on manifolds with integrable $G_2$ structure. We also prove, following \cite{carrion1998generalization,reyes1993some},  that a canonical $G_2$ cohomology can be constructed for any system $(Y,V)$, where $Y$ is a manifold with integrable $G_2$ structure, and $V$ and instanton bundle. To achieve this, we define a new operator $\check \dd_A$, and show that this gives rise to an elliptic complex. In section \ref{sec:geommod} we reproduce known results for the infinitesimal moduli space of $G_2$ manifolds, and in particular how the moduli are mapped to the canonical $G_2$ cohomology group $H_{\check \dd_{\theta}}^1(Y,TY)$. Finally, in section \ref{sec:bundmod}, we study the variations of the instanton bundle $V$, and the combined system $(Y,V)$. We show that the moduli space corresponds to 
\[  
H^1_{\check \dd_A}(Y, {\rm End}(V))\oplus {\rm ker}(\check{\cal F})\subset H^1_{\check \dd_A}(Y, {\rm End}(V))\oplus H_{\check \dd_{\theta}}^1(Y,TY)~,\]
where elements in $H^1_{\check \dd_A}(Y, {\rm End}(V))$ correspond to bundle moduli and the geometric moduli are restricted to lie in the kernel of the $G_2$ Atiyah map $\check{\cal F}$. This result is also discussed from the perspective of extension bundles.

%\newpage

\section{Manifolds with $G_2$ structure}
\label{sec:g2struc}

In this section, we recall relevant facts about manifolds with $G_2$ holonomy. Our discussion is brief, and the reader is referred to \cite{MR916718,bonan66,FerGray82,Hitchin:2000jd,joyce2000,Bryant:2005mz} for further details. Let $Y$ be a 7-dimensional manifold. A $G_2$ structure on $Y$ exists when the first and second Stiefel-Whitney classes are trivial, that is when $Y$ is orientable and spin. When this is the case, $Y$ admits a nowhere-vanishing Majorana spinor $\eta$. Equivalently, $Y$ has a non-degenerate,   associative 3-form $\varphi$, constructed as a spinor bilinear:
\[
\varphi_{abc} =-i  \eta^{\dagger} \gamma_{abc} \eta~.
\]
Here $\gamma_{abc}$ is an antisymmetric product of three 7-dimensional $\gamma$ matrices, that we take to be Hermitian and purely imaginary. We note that the three-form $\varphi$ is positive, as is required to define a $G_2$ structure \cite{joyce2000}. We will often refer to $\varphi$ as a $G_2$ structure. $Y$ has $G_2$ holonomy when $\eta$ is covariantly constant with respect to the Levi--Civita connection:
\be
\label{eq:covconstspin}
\nabla \eta = 0
\ee
 or equivalently when $\varphi$ is closed and co-closed.

The form $\varphi$ determines a Riemannian metric $g_\varphi$
on $Y$ by
\begin{equation}
6 g_\varphi(x, y)\, \dd {\rm vol}_\varphi
= (x\lrcorner\varphi)\wedge(y\lrcorner\varphi)\wedge\varphi~,
\label{eq:g2metric}
\end{equation}
for all vectors $x$ and $y$ in $\Gamma(TY)$.  In components this means
\begin{equation}
g_{\varphi\, ab} = \frac{\sqrt{\det g_\varphi}}{3!\, 4!}\, 
\varphi_{a c_1 c_2}\, \varphi_{b c_3 c_4}\, \varphi_{c_5 c_6 c_7}\,
\epsilon^{c_1\cdots c_7}
=  \frac{1}{4!}\, 
\varphi_{a c_1 c_2}\, \varphi_{b c_3 c_4}\, 
\psi^{c_1 c_2 c_3 c_4}~,
\label{eq:g2metricab}
\end{equation}
where
\[ \psi = *\varphi ~,\]
which in terms of spinors corresponds to $\psi_{abcd}= \eta^{\dagger} \gamma_{abcd} \eta$, and 
\[ \dd x^{a_1\cdots a_7} = \sqrt{\det g_\varphi}
\ \epsilon^{a_1\cdots a_7}\, \dd {\rm vol}_\varphi~.\]
With respect to this metric, the 3-form $\varphi$, and hence its Hodge dual $\psi$, are  normalised so that
\[ \varphi\wedge *\varphi = ||\varphi||^2\, \dd{\rm vol}_\varphi
~, \qquad ||\varphi||^2= 7~,\]
that is
\[ \varphi\lrcorner\varphi = \psi\lrcorner\psi = 7
~.\]

\subsection{Decomposition of forms}
\label{sec:formdec}

The existence of a $G_2$ structure $\varphi$ on $Y$ determines a decomposition of differential forms on $Y$ into irreducible representations of $G_2$.  This decomposition changes when one deforms the $G_2$ structure. 

Let $\Lambda^k(Y)$ be the space of $k$-forms on $Y$ and $\Lambda_p^k(Y)$ be the subspace of $\Lambda^k(Y)$ of $k$-forms which transform in the $p$-dimensional irreducible representation of $G_2$.   We have the following decomposition for each $k= 0,1, 2, 3$:\footnote{Note that $T^*Y \cong TY$ only as vector spaces.}
\begin{align*}
\Lambda^0 &= \Lambda_1^0
~,\\
\Lambda^1 &= \Lambda_7^1 = T^*Y \cong TY
~,\\
\Lambda^2 &= \Lambda_7^2\oplus \Lambda_{14}^2
~,\\
\Lambda^3 &= \Lambda_1^3\oplus\Lambda_7^3\oplus\Lambda_{27}^3
~.
\end{align*}
The decomposition for $k = 4, 5, 6, 7$ follows from the Hodge dual for $k = 3, 2, 1, 0$ respectively. For a form of a given degree, the decomposition into $G_2$ representations is obtained using contractions and wedge products with $\varphi$, see  \cite{MR916718}. A comprehensive discussion will also appear in \cite{delaOssa}.

\subsection{Torsion classes}

Decomposing into representations of $G_2$ the exterior derivatives of $\varphi$ and $\psi$  we have 
\begin{align}
\dd_7\varphi &= \tau_0\psi + 3\, \tau_1\wedge\varphi + *_7\tau_3~,
\label{eq:Intphi}\\
\dd_7\psi &= 4\, \tau_1\wedge\psi + *\tau_2~,\label{eq:Intpsi}
\end{align}
where the $\tau_i\in \Lambda^i(Y)$ are the {\it torsion classes}, which are {\it uniquely} determined by the $G_2$-structure 
$\varphi$ on $Y$ \cite{FerGray82}. We note that $\tau_2\in \Lambda^2_{14}$ and that $\tau_3\in \Lambda^3_{27}$. 
A $G_2$ structure for which 
\[ 
\tau_2 = 0~,
\]
will be called an {\it integrable} $G_2$ structure, using the parlance of Fern\'andez-Ugarte \cite{fernandez1998dolbeault}. The manifold $Y$ has {\it $G_2$ holonomy} if and only if all torsion classes vanish.

\subsection{Cohomologies on $G_2$ structure manifolds}
In this section, we recall different cohomologies that are of relevance for $G_2$ holonomy manifolds. In fact, a large part of our discussion is valid for a larger class of $G_2$ structure manifolds, namely the integrable ones. When we can, we will state our results for this larger class of manifolds, of which the $G_2$ holonomy manifolds form a subclass.

\subsubsection{de Rham cohomology}
For completeness, and to state our notation, let us first discuss the de Rham complex. As above, $\Lambda^p(Y)$ denotes the bundle of $p$-forms on $Y$. The exterior derivative 
\be
\dd: \Lambda^p(Y)\rightarrow\Lambda^{p+1}(Y)
\ee
maps $p$-forms to $p+1$ forms:
\be
\dd \omega = \sum_{j,I} \frac{\partial \omega_I}{\partial x_j} \dd x_j \wedge \dd x^I \; .
\ee
Since $\dd^2 = 0$, the sequence
\be
0\xrightarrow{\dd}\Lambda^0(Y)\xrightarrow{\dd}\Lambda^1(Y) .... \xrightarrow{\dd}\Lambda^d(Y)\xrightarrow{\dd}0 \; .
\ee
forms a complex. We show in detail in appendix \ref{app:Elliptic} that this de Rham complex  is elliptic. As a consequence, the de Rham cohomology groups
\be
H^p(Y) = {\rm ker}  (\dd_p) / {\rm im}  (\dd_{p-1})
\ee
are finite-dimensional for compact $Y$. Finally, using the wedge product, we see that $H^*(Y)$ is endowed with a natural ring structure, {\it cf.}~Theorem \ref{tm:ringY} below.

\subsubsection{The canonical $G_2$ cohomology}
\label{subsec:checkD}

We now turn to the Dolbeault complex for manifolds with an integrable $G_2$ structure which was first constructed in \cite{reyes1993some} and \cite{fernandez1998dolbeault}.  In these references, a differential operator $\cd$ acting on a sub-complex of the de Rham complex of $Y$, is defined in analogy with a Dolbeault operator on a complex manifold. 

\begin{definition} The differential operator $\cd$ is defined by 
the maps 
\begin{align*}
\cd_0&: \Lambda^0(Y)\rightarrow\Lambda^1(Y)~,\qquad\qquad  
\cd_0f = \dd f~, \qquad\quad f\in\Lambda^0(Y)
~,\\
\cd_1&: \Lambda^1(Y)\rightarrow\Lambda_7^2(Y)~,\qquad\qquad 
\cd_1\alpha = \pi_7(\dd \alpha)
~, \quad \alpha\in\Lambda^1(Y)
~,\\
\cd_2 &: \Lambda_7^2(Y)\rightarrow\Lambda_1^3(Y)~,\qquad\qquad
\cd_2\beta = \pi_1(\dd\beta)~, \quad \beta\in\Lambda_7^2(Y)~.
\end{align*}
That is, 
\begin{equation*}
\cd_0=\dd~,\quad\cd_1=\pi_7\circ\dd~,\quad\cd_2=\pi_1\circ\dd ~.
\end{equation*}
\end{definition}

Consider the following lemma
\begin{lemma}
\label{lem:dontwo14}
Let $Y$ be an integrable $G_2$ holonomy manifold and $\beta\in\Lambda^2_{14}(Y)$. Then 
\[\dd\beta\in\Lambda^3_7(Y)\oplus\Lambda^3_{27}(Y)~.\]
\end{lemma}
\begin{proof}
Consider
\begin{align*}
0= \dd(\beta\wedge\psi) = \dd\beta\wedge\psi + \beta\wedge\dd\psi
\end{align*}
Hence
\[
\dd\beta\wedge\psi = - \beta\wedge\dd\psi = - 4\, \beta\wedge\tau_1\wedge\psi = 0~.
\]
Therefore the result follows.

\end{proof}

We then have the following theorem:
\begin{theorem}\label{prop:dolbcomplex}
Let $Y$ be a manifold with a $G_2$ structure. Then
\begin{equation}
\label{eq:dolb}
0\rightarrow\Lambda^0(Y)\xrightarrow{\cd}\Lambda^1(Y)\xrightarrow{\cd}\Lambda^2_7(Y)\xrightarrow{\cd}\Lambda^3_1(Y)\rightarrow0
\end{equation}
is a differential complex, i.e. $\cd^2=0$ if and only if the $G_2$ structure is integrable, that is, $\tau_2 = 0$ .
\end{theorem}
\begin{proof}
Let $f\in \Lambda^0(Y)$.  Then
\[ \cd^2 f = \pi_1\dd(\dd f) = 0~.\]
Consider $\alpha\in\Lambda^1(Y)$.  In this case
\[ \cd^2\alpha = \pi_1\big(\dd(\pi_7(\dd\alpha))\big)
= \pi_1\big(\dd( \dd\alpha - \pi_{14}(\dd\alpha))\big)
= -  \pi_1\big(\dd(\pi_{14}(\dd\alpha))\big)
~.
\]
Hence
\[  \cd^2\alpha = 0\quad{\rm iff}\quad \dd(\pi_{14}(\dd\alpha))\in \Lambda_7^3\oplus\Lambda_{14}^3
\quad{\rm iff}\quad \dd(\pi_{14}(\dd\alpha))\wedge\psi = 0~,
\]
for all $\alpha\in \Lambda^1(Y)$.  We have
\[ \dd(\pi_{14}(\dd\alpha))\wedge\psi 
= \dd(\pi_{14}(\dd\alpha)\wedge\psi) - (\pi_{14}(\dd\alpha))\wedge\dd\psi
= - (\pi_{14}(\dd\alpha))\wedge*\tau_2
~.\]
Therefore
\[ \cd^2\alpha = 0\quad{\rm iff}\quad (\pi_{14}(\dd\alpha))\wedge*\tau_2 = 0~,\]
for all $\alpha\in \Lambda^1(Y)$.
This can only hold true iff $\tau_2=0$.

\end{proof}

We denote the complex \eqref{eq:dolb} by $\check\Lambda^*(Y)$. It should be mentioned that the complex \eqref{eq:dolb} is actually an elliptic complex \cite{reyes1993some}. We give a proof of this in appendix \ref{app:Elliptic}. We denote by $H_{\check \dd}^*(Y)$ the corresponding cohomology ring, which is often referred to as the canonical $G_2$-cohomology of $Y$ \cite{fernandez1998dolbeault}. 

One curiosity to note about $\cd$ is that in contrast to the familiar differentials like the de Rham operator $\dd$ or the Dolbeault operators $\bar\partial$ and $\partial$, $\cd$ does not generically satisfy a Poincare lemma. To see why, consider $\alpha\in\check\Lambda^1(Y)=\Lambda^1(Y)$. If there was a Poincare lemma, then $\cd\alpha=0$ would imply that $\alpha=\cd f=\dd f$ for some locally defined function $f$. But then we would have $\dd\alpha=0$, which is not true in general. In other words the complex \eqref{eq:dolb} is not locally trivial. Hence, it becomes harder to define a notion of sheaf cohomology for $\cd$.

Note that we can endow $H_{\check \dd}^*(Y)$ with a natural ring structure. Indeed, we have the following theorem
\begin{theorem}
\label{tm:ringY}
The wedge product induces a well-defined ring structure on the cohomology $H_{\check \dd}^*(Y)$. The corresponding symmetric product is denoted by
\begin{equation*}
(\:,\:)\::\: H_{\check \dd}^p(Y)\times H_{\check \dd}^q(Y)\rightarrow H_{\check \dd}^{p+q}(Y)\:,
\end{equation*}
and is given by, for $\alpha\in H_{\check \dd}^p(Y)$ and $\beta\in H_{\check \dd}^q(Y)$,
\begin{equation*}
(\alpha,\beta)=\pi_i(\alpha\wedge\beta)\:.
\end{equation*}
where $\pi_i$ denotes the appropriate projection onto the correct subspace $\Lambda^{p+q}_i(Y)$ of $\Lambda^{p+q}(Y)$.
\end{theorem}
\begin{proof}
The proof of this theorem is very similar in spirit to the proof of Theorem \ref{tm:ringA} below. One needs to show that if $\alpha$ and $\beta$ are $\cd$-closed, then $(\alpha,\beta)$ is $\cd$-closed. Also, in order to be a well-defined product, if either $\alpha$ or $\beta$ are $\cd$-exact, then the product should also be exact. We leave this as an exercise for the reader. 

\end{proof}

\subsubsection{A canonical $G_2$ cohomology for $TY$}
\label{subsec:thetaconnection}
 
 In the following, and in the accompanying paper \cite{delaOssa}, we will  discover that deformations of $G_2$ holonomy manifolds can be understood by means of a connection $\dth$ on the tangent bundle $TY$.    In anticipation of these results, in this subsection we define this connection and include a number of properties.   
   
 Let $\Delta^a$ be a $p$-form with values in $TY$, that is $\Delta\in\Lambda^p(TY)$. Let $\dth$ be a connection on $TY$ defined by
 \[ \dth \Delta^a = \dd\Delta^a + \theta_b{}^a\wedge \Delta^b~,\]
 where the connection one form $\theta_b{}^a$ is given by
 \[ \theta_b{}^a = \Gamma_{bc}{}^a \dd x^c~,\]
 and $\Gamma$ are the connection symbols of a metric connection $\nabla$ on $Y$ which is compatible with the $G_2$ structure, that is
 \[ \nabla\varphi = 0~,\qquad \nabla\psi = 0 ~.\]
On $G_2$ holonomy manifolds, this connection is unique, and corresponds to the Levi--Civita connection. Thus, we have 
\be
\label{eq:dthetahol}
 \dth \Delta_t^a =  \dd \Delta_t^a + \theta_b{}^a\wedge \Delta_t^b = \nabla^{LC}_b\, \Delta_{t\,c}{}^a\,\dd x^{bc}~.
 \ee
Note that this implies that the connection $\dth$ is metric.

 Given the connection $\dth$ on $TY$ defined in this subsection, one can define the operator $\check \dth$ as in definition \ref{def:checkdA}, and a complex $\check\Lambda^*(Y,TY)$ as in equation \eqref{eq:dolb}.   We then have:
\begin{theorem}\label{prop:dolbcomplex2}
Let $Y$ be a manifold with integrable $G_2$ structure. Then
\begin{equation}
\label{eq:dthcompl}
0\rightarrow\Lambda^0(TY)\xrightarrow{\check \dth}\Lambda^1(TY)\xrightarrow{\check \dth}\Lambda^2_7(TY)\xrightarrow{\check \dth}\Lambda^3_1(TY)\rightarrow0
\end{equation}
is a differential complex, i.e. $\check \dth^2=0$ if and only if  $\check R(\theta)$ is an instanton, {\it i.e.} $\check R(\theta)_a{}^b \w \psi = 0$.
\end{theorem}
\begin{proof}
 We omit this proof, since it is similar to the proofs of Theorems \ref{prop:dolbcomplex} and  \ref{th:dacheck}.  \end{proof}
 
On a $G_2$ holonomy manifold, Theorem \ref{prop:dolbcomplex2} always holds, since the curvature 
 \[ R(\theta)_a{}^b = \dd\theta_a{}^b + \theta_c{}^b\wedge\theta_a{}^c~,\]
 equals the curvature of the Levi-Civita connection $\nabla$: 
\[
(R(\theta)_a{}^b)_{cd} = \partial_c \Gamma_{ad}{}^b + \Gamma_{ec}{}^b\wedge\Gamma_{ad}{}^e = 
\partial_c \Gamma_{da}{}^b + \Gamma_{ce}{}^b\wedge\Gamma_{da}{}^e
=
(R(\nabla)_a{}^b)_{cd}~.
\]
Consequently, we may denote the curvature for both connections by $R$. 
Moreover, integrability of the spinorial constraint \eqref{eq:covconstspin} for $G_2$ holonomy implies that $\nabla$ is an instanton 
\[
[\nabla_n, \nabla_p] \eta =0  \iff 
R_{np\,ab}\gamma^{ab} \eta = 0 \iff
R_a{}^b \w \psi = 0
\, .
\]
It thus follows that $G_2$ holonomy implies that $\theta$ is an instanton. As a consequence, $TY$ is an instanton bundle with connection $\theta$. We will discuss instanton bundles in complete generality in next section, and will prove that the complex \eqref{eq:dthcompl} is elliptic and that the associated cohomology groups $H_{\check \dd_{\theta}}^p(Y,TY)$ are finite-dimensional (if $Y$ is compact). 

%\newpage

\section{Instanton bundles on manifolds with integrable $G_2$ structure}
\label{sec:instbund}

In this section, we discuss vector bundles with an instanton connection over manifolds with $G_2$ structure.  Higher-dimensional instanton equations generalise the self-dual Yang-Mills equations in four dimensions, and were first constructed in \cite{1980PhR....66..213E,CORRIGAN1983452,WARD1984381}. The instanton condition can be reformulated as a $G_2$ invariant constraint \cite{2000math.....10015T,donaldson1998gauge, 2009arXiv0902.3239D, Harland:2009yu,2010JHEP...10..044B,Harland:2011zs,Ivanova:2012vz,Bunk:2014coa,Bunk:2014kva,Haupt:2014ufa,2015arXiv151104928H}, and explicit solutions to the instanton condition on certain $G_2$ manifolds are also known \cite{2011arXiv1109.6609W,2014JGP....82...84C}. Here, we show that the $G_2$ instanton condition is implied by a supersymmetry constraint in string compactifications, and that it, in turn, implies the Yang--Mills equations as an equation of motion of the theory. In the second part of this chapter, we define an elliptic Dolbeault cohomology on $G_2$ instanton bundles, which we will use in the subsequent discussion of the infinitesimal moduli space of $G_2$ manifolds with instanton bundles.

\subsection{Instantons and Yang--Mills equations}

Let $Y$ be a $d$-dimensional real Riemannian manifold and let $V$ be a vector bundle on $Y$ with connection $A$.
Suppose $Y$ has a $G$-structure and that $Q$ is a $G$-invariant four-form on $Y$.  The connection $A$ on $V$ is an instanton
if for some real number $\nu$ (typically $\nu = \pm 1$), the curvature $F= \dd A + A\wedge A$ satisfies (see e.g.~\cite{Harland:2011zs})
\begin{equation}
 F\wedge *Q = \nu\, *F~.\label{eq:inst1}
 \end{equation}
In fact, taking the Hodge dual, 
equation \eqref{eq:inst1} is
\begin{equation}
F\lrcorner Q = \nu\, F~.\label{eq:inst2}
\end{equation} 
In the case when $G=G_2$ and $d= 7$, the $G_2$-invariant four-form is $Q=\psi = *\varphi$, so
\[ F\wedge\varphi = - *F \iff F\lrcorner \psi = - F~,\]
where we have taken the Hodge dual in the second equality. This is the condition that $F\in \Lambda^2_{14}(Y,{\rm End}(V))$ and it is equivalent to
\begin{equation}
 F\wedge\psi = 0~.\label{eq:G2inst}
 \end{equation}

An instanton is supposed to satisfy the Yang-Mills equation, which in our case,  appears as an equation of motion of the superstring theory.  We will review how this works for the general $d$-dimensional case with non-zero torsion, specialising at the end of this section to $d=7$ and $G_2$ holonomy. Note also that the instanton equation is implied from the vanishing of the supersymmetric variation of the gaugino
\[ F_{mn}\,\gamma^{mn}\, \eta = 0~,\]
 whenever we are considering compactifications which preserve some supersymmetry (here $\eta$ is a nowhere vanishing globally well defined spinor which defines the $G$-structure on Y, {\it cf.} section \ref{sec:g2struc}).  Hence the Yang-Mills equation (as an equation of motion) is satisfied if this supersymmetry condition (as an instanton) is satisfied.

To see that equation \eqref{eq:inst1} satisfies the Yang-Mills equation, we begin by taking the exterior derivative of equation \eqref{eq:inst1}
\begin{equation}
\dd F \wedge *Q + F\wedge \dd*Q= \nu\, \dd *F~.\label{eq:dinst}
\end{equation}
Using the Bianchi identity for $F$  
\[ \dd_A F = \dd F + A\wedge F - F\wedge A =0~,\]
on the first term of the left hand side of equation \eqref{eq:dinst} we have 
\begin{equation*}
 \dd F \wedge *Q = (- A\wedge F + F\wedge A) \wedge *Q 
= \nu\, (- A \wedge *F + (-1)^d\, *F\wedge A)~.
\end{equation*}
Plugging this back into equation \eqref{eq:dinst} and rearranging 
we find
\begin{equation}
\nu\,  \dd_A*F = F\wedge \dd*Q~,
\label{eq:predinst}
\end{equation}
where
\[ \dd_A\beta = \dd\beta + A\wedge\beta - (-1)^k\, \beta\wedge A~,\]
for any $k$-form $\beta$ with values in ${\rm End}(V)$.

Recall that in $d$-dimensions, for any $k$-form with values in ${\rm End}(V)$
\begin{align*}
\dd^\dagger_A \beta &= (-1)^{ dk + d +1} * \dd_A *\beta \\[3pt]
&= \dd^\dagger\beta 
+ (-1)^{dk + d + 1}\, * (A\wedge*\beta + (-1)^{d + k + 1} *\beta\wedge A)~.
\end{align*}
Therefore, taking the Hodge dual of \eqref{eq:predinst} 
we find
\begin{equation}
\nu\,\dd^\dagger_A F= F\lrcorner \dd^\dagger Q~,\label{eq:YM}
\end{equation}
which should then be the Yang-Mills equation when there is non-vanishing torsion. In the $G_2$ holonomy case, we have that $Q=\psi$ is closed, by which we conclude that
\be
\label{eq:YMholo}
\dd^\dagger_A F = 0 \quad\quad \mbox{($G_2$ holonomy)}\; .
\ee
This is in fact the equation of motion for the dilaton in fluxless $\cN=1$ supersymmetric compactifications of the heterotic string, as can be seen using the identity \eqref{eq:ddaggeralpha} and comparing with equation (A.4d) in \cite{Gauntlett:2003cy}. In a similar fashion, one may show that \eqref{eq:YM} is indeed the 
equation of motion for the dilaton when there is non-vanishing torsion (as discussed in \cite{Gauntlett:2003cy} this is requires that $Y$ permits generalised calibrations, which relate the $H$-flux to $\dd^{\dagger}Q$).

\subsection{A  canonical $G_2$ cohomology for instanton bundles}
\label{sec:BundleCohomology}

Let us now construct a Dolbeault-type cohomology that generalizes the canonical $G_2$ cohomology of $Y$ to a vector bundle $V$ over $Y$, as was first done in \cite{carrion1998generalization,reyes1993some}. We assume that the connection $A$ on $V$ is an instanton, so that its curvature satisfies
\begin{equation}
\label{eq:holbund}
\psi\wedge F=0\:,
\end{equation}
or, equivalently, $F\in\Lambda^2_{14}(Y,\End(V))$.
We will state all results of this section in the most general terms, namely for integrable $G_2$ structures and for forms with values in a vector bundle $E$, where the bundle $E$ can be $V$, $V^*$, $\End(V) = V \otimes V^*$, or any other sum or product of these bundles.
We note first that Lemma \ref{lem:dontwo14} readily generalises to the exterior derivative $\dd_A$.

\begin{lemma}\label{lem:dAontwo14}
Let $\beta$ be a two form with values in a vector bundle $E$ defined above.  Let $A$ be any connection on $V$.
If $\beta\wedge\psi=0$, that is if $\beta\in \Lambda_{14}^2(Y,E)$, then
\[ \dd_A\beta\in \Lambda_{7}^3(Y,E)\oplus\Lambda_{27}^3(Y,E)~.\]
\end{lemma}
\begin{proof}
Consider
\begin{align*}
0= \dd_A(\beta\wedge\psi) = \dd_A\beta\wedge\psi + \beta\wedge\dd\psi
\end{align*}
Hence
\[
\dd_A\beta\wedge\psi = - \beta\wedge\dd\psi = - 4\, \beta\wedge\tau_1\wedge\psi = 0~.
\]
The result follows.

\end{proof}

We now define the following differential operator

\begin{definition}\label{def:checkdA}
The maps $\cd_{iA}, i=0,1,2$  are given by
\begin{align*}
\cd_{0A}&: \Lambda^0(Y,E)\rightarrow\Lambda^1(Y,E)~,\qquad  
\qquad 
\cd_{0A}f = \dd_A f~, \qquad\quad f\in\Lambda^0(Y,E)
~,\\
\cd_{1A}&: \Lambda^1(Y,E)\rightarrow\Lambda^2_7(Y,E)~,\qquad\qquad 
\cd_{1A}\alpha = \pi_7(\dd_A \alpha)
~, \quad \alpha\in\Lambda^1(Y,E)
~,\\
\cd_{2A} &: \Lambda^2(Y,E)\rightarrow\Lambda^3_1(Y,E)~,\qquad\qquad
\cd_{2A}\beta = \pi_1(\dd_A\beta)~, \quad \beta\in\Lambda_7^2(Y,E)~.
\end{align*}
where the $\pi_i$'s denote projections onto the corresponding subspace.
\end{definition}
It is easy to see that these operators are well-defined under gauge transformations. We then have:
\begin{theorem}
\label{th:dacheck}
Let $Y$ be a seven dimensional manifold with a $G_2$ structure. The complex
\begin{equation}
\label{eq:dolbV}
0\rightarrow\Lambda^0(Y,E)\xrightarrow{\cd_A}\Lambda^1(Y,E)\xrightarrow{\cd_A}\Lambda^2_7(Y,E)\xrightarrow{\cd_A}\Lambda^3_1(Y,E)\rightarrow0
\end{equation}
is a differential complex, i.e. $\cd_A^2=0$, if and only if the connection $A$ on $V$ is an instanton and the manifold has an integrable $G_2$ structure. We shall denote the complex \eqref{eq:dolbV} $\check\Lambda^*(Y,E)$, where $E$ is one of the bundles discussed above.
\end{theorem}
\begin{proof}
Let $f\in\Lambda^0(Y,E)$. Then
\begin{equation*}
\cd_A^2f=\pi_7(\dd_A^2f)=(\pi_7F)\, f~.
\end{equation*}
Hence
\[ \cd_A^2f=0~~\forall \, f\in\Lambda^0(Y, V)\quad{\rm iff}\quad  F\wedge\psi = 0~,\]
{\it i.e.}~the connection $A$ on the bundle $V$ is an instanton.
Now, consider $\alpha\in\Lambda^1(Y,E)$.  In this case
\[ \cd_A^2\alpha = \pi_1\big(\dd_A(\pi_7(\dd_A\alpha))\big)
= \pi_1\big(\dd_A( \dd_A\alpha - \pi_{14}(\dd_A\alpha))\big)
=  \pi_1\big(F\wedge\alpha - \dd_A(\pi_{14}(\dd_A\alpha))\big)
~,
\]
where we recall that we find the singlet representation of a three-form by contracting with $\varphi$, or wedging with $\psi$. Thus, the first term vanishes,  since $F$ is an instanton. Hence
\[  \cd_A^2\alpha = 0\quad{\rm iff}\quad \dd_A(\pi_{14}(\dd_A\alpha))\wedge\psi = 0~,
\]
for all $\alpha\in \Lambda^1(Y)$.  We have
\[ \dd_A(\pi_{14}(\dd_A\alpha))\wedge\psi 
= \dd_A(\pi_{14}(\dd_A\alpha)\wedge\psi) - (\pi_{14}(\dd\alpha))\wedge\dd\psi
= - (\pi_{14}(\dd_A\alpha))\wedge*\tau_2
~.\]
Therefore
\[ \cd^2\alpha = 0\quad{\rm iff}\quad (\pi_{14}(\dd_A\alpha))\wedge*\tau_2 = 0~,\]
for all $\alpha\in \Lambda^1(Y,E)$.
This holds true iff $\tau_2=0$.

\end{proof}
Note that  by a similar argument as given for the complex \eqref{eq:dolb} in appendix \ref{app:Elliptic}, it follows that the complex \eqref{eq:dolbV} is elliptic, as was also shown in \cite{carrion1998generalization}. As a consequence, the corresponding cohomology groups are of finite dimension, provided that $Y$ is compact.

Finally, we prove the following theorem, which generalises Theorem \ref{tm:ringY}: 
\begin{theorem}
\label{tm:ringA}
We have a ring structure on the cohomology $H_{\check \dd_{A}}^*(Y,\End(V))$,
\begin{equation*}
\pi_i[\:,\:]\::\: H_{\check \dd_{A}}^p(Y,\End(V)))\times H_{\check \dd_{A}}^q(Y,\End(V)))\rightarrow H_{\check \dd_{A}}^{p+q}(Y,\End(V))\:,
\end{equation*}
where $\pi_i$ denotes the appropriate projection.
\end{theorem}
\begin{proof}
The cases $\{p=0,q=n\}$ for $n=\{0,1,2,3\}$ are easily proven. For the case $p=q=1$, note that if $\alpha_{1,2}\in\Lambda^1(Y,\End(V))$ are
are $\cd_A$-closed, then 
\begin{equation*}
\cd_A\pi_7([\alpha_1,\alpha_2])=0\:.
\end{equation*}
Indeed, we have
\begin{equation*}
\dd_A([\alpha_1,\alpha_2])=\dd_A\pi_7([\alpha_1,\alpha_2])+\dd_A\pi_{14}([\alpha_1,\alpha_2])\:.
\end{equation*}
Wedging this with $\psi$, using that $\alpha_{1,2}$ are $\cd_A$-closed, and applying Lemma \ref{lem:dAontwo14} on the last term after the last equality, the result follows. Note also that if e.g. $\alpha_2$ is trivial, that is $\alpha_2=\dd_A\epsilon_a$, we get
\begin{equation*}
[\alpha_1,\alpha_2]\wedge\psi=[\alpha_1,\dd_A\epsilon_a]\wedge\psi=-\dd_A([\alpha_1,\epsilon_a])\wedge\psi\:,
\end{equation*}
and so $\pi_7[\alpha_1,\alpha_2]=-\pi_7(\dd_A[\alpha_1,\epsilon_a])=-\cd_A[\alpha_1,\epsilon_a]$. We thus find a well-defined product on the level of one-forms. By symmetry of the product, the only case left to consider is $\{p=1,q=2\}$. We let $\alpha\in\Lambda^1(Y,\End(V))$ and $\beta\in\Lambda^2_{7}(Y,\End(V))$. Clearly
\begin{equation*}
\cd_A[\alpha,\beta]=0\:.
\end{equation*}
We only need to show that the product is well-defined. That is, let $\alpha=\cd_A\epsilon=\dd_A\epsilon$. We then have
\begin{equation*}
\pi_1[\alpha,\beta]=\pi_1[\dd_A\epsilon,\beta]=\cd_A[\epsilon,\beta]-\pi_1[\epsilon,\dd_A\beta]=\cd_A[\epsilon,\beta]\:,
\end{equation*}
as $\cd_A\beta=0$. Similarly, let $\beta=\cd_A\gamma=\pi_7\dd_A\gamma$ for $\gamma\in\Lambda^1(Y,\End(V))$. Then $\beta=\dd_A\gamma+\kappa$, where $\kappa\in\Lambda^2_{14}(Y,\End(V))$. We then have
\begin{equation*}
\psi\wedge[\alpha,\beta]=\psi\wedge[\alpha,\dd_A\gamma+\kappa]=\psi\wedge[\alpha,\dd_A\gamma]=-\psi\wedge\dd_A[\alpha,\gamma]\:,
\end{equation*}
where we have used that $\psi\wedge\dd_A\alpha=0$. Hence
\begin{equation*}
\pi_1[\alpha,\beta]=-\cd_A[\alpha,\gamma]\:.
\end{equation*}
It follows that the product is well defined. This concludes the proof.

\end{proof}
We will drop the projection $\pi_i$ from the bracket when this is clear from the context. As a corollary of Theorem \ref{tm:ringA} it is easy to see that the complex $\check\Lambda^*(Y,\End(V))$ forms a differentially graded Lie algebra. That is, there is a bracket
\begin{equation*}
[\cdot,\cdot]\::\;\;\;\check\Lambda^p(Y,\End(V))\otimes\check\Lambda^q(Y,\End(V))\;\;\;\rightarrow\;\;\;\check\Lambda^{p+q}(Y,\End(V))\:,
\end{equation*}
which is simply inherited from the Lie-bracket of $\End(V)$. As a result, this bracket also satisfies the Jacobi identity. Moreover, following similar arguments to that of the proof of Theorem \ref{tm:ringA}, it is easy to check that for $x\in\Lambda^p(Y,\End(V))$ and $y\in\Lambda^q(Y,\End(V))$ we have
\begin{equation}
\cd_A[x,y]=[\cd_Ax,y]+(-1)^p[x,\cd_Ay]\:.
\end{equation}
It follows that $\check\Lambda^*(Y,\End(V))$ forms a differentially graded Lie algebra. We will return to this in section \ref{sec:higherdef} when discussing higher order deformations of the bundle.

\subsubsection{Hodge Theory}

We now want to consider the Hodge-theory of the complex \eqref{eq:dolbV}. To do so, we need to define an adjoint operator of $\cd_A$. We have the usual inner product on forms on $Y$,
\begin{equation*}
(\alpha,\beta)=\int_Y\alpha\wedge*\beta
\end{equation*}
for $\{\alpha,\beta\}\in\Lambda^*(Y)$. Note that forms in different $G_2$ representations are orthogonal with respect to the inner product. We want to extend this to include an inner product on forms valued in $V$ and $\End(V)$. In the case of endomorphism bundles, we can make use 
of the trace
\begin{equation*}
(\alpha,\beta)=\int_Y\tr\:\alpha\wedge*\beta\:,
\end{equation*}
for $\{\alpha,\beta\}\in\Lambda^*(Y,\End(V))$. For a generic vector bundle $E$, we must specify a metric $G_{xy}\in\Lambda^0\left({\rm Sym}(E^*\otimes E^*)\right)$, in order to define the inner product
\begin{equation} \label{eq:inprod}
(\alpha,\beta)=\int_Y\alpha^x\wedge*\beta^y\,G_{xy}\:,
\end{equation}
for $\{\alpha^x,\beta^y\}\in\Lambda^*(Y,E)$.  As in the case of endomorphism bundles, we may choose a trivial metric $\delta_{xy}$, but other choices may be more natural. In order to simplify our analysis, we will keep the metric $G_{xy}$ arbitrary, but require it to be parallel to $\cd_A$: 
\begin{equation*}
\cd_AG_{xy}=\dd_AG_{xy}=0\:.
\end{equation*}
In the case of complex structures, this would be a Hermiticity condition that uniquely specifies the Chern-connection. For $G_2$ structures, things are a bit more subtle, and we will return to this discussion in the companion paper \cite{delaOssa}. Note however that when $E=TY$, we can use the canonical metric $g_{\varphi}$ in the inner product \eqref{eq:inprod}. In the case when $Y$ has $G_2$ holonomy, the connection on $TY$ will simply be the Levi-Civita connection, which is metric. 

Having specified an inner product on $E$, we would now like to construct the adjoint operators of $\cd_A$ and also use these to construct elliptic Laplacians. We have the following proposition
\begin{proposition}\label{prop:prop2}
With respect to the above inner-product, and with $G_{xy}$ is parallel to $\cd_A$, the adjoint of $\cd_A$ is given by
\begin{equation*}
\cd_A^\dagger=\pi\circ\dd_A^\dagger\:,\;\;\;\textrm{where}\;\;\;\dd_A^\dagger=-*\dd_A*\:,
\end{equation*}
Here $\pi$ denotes the appropriate projection for the degree of the forms involved.
\end{proposition} 
\begin{proof}
Consider $\alpha\in\Lambda_7^2(Y,E)$ and $\gamma\in\Lambda_1^3(Y,E)$. Using definition \ref{def:checkdA}, the inner product \eqref{eq:inprod}, and the orthogonality of forms in different $G_2$ representations, we then compute
\begin{equation*}
(\alpha,\cd_A^\dagger\gamma)=(\alpha,\pi_7\circ\dd_A^\dagger\gamma)=(\alpha,\dd_A^\dagger\gamma)=(\dd_A\alpha,\gamma)=(\cd_A\alpha,\gamma)\:.
\end{equation*}
The cases for forms of other degrees are similar.
\end{proof}

Using a parallel metric $G_{xy}$, we can then construct the Laplacian
\begin{equation*}
\check\Delta_A=\cd_A\cd_A^\dagger+\cd_A^\dagger\cd_A\:.
\end{equation*}
With this Laplacian, we now prove a Hodge-theorem of the following form
\begin{theorem}
\label{tm:hodge}
The forms in the differential complex \eqref{eq:dolbV} have an orthogonal decomposition
\begin{equation*}
\check\Lambda^*(Y,E)={\rm Im}(\cd_A)\oplus{\rm Im}(\cd_A^\dagger)\oplus{\rm ker}(\check\Delta_A)\:.
\end{equation*}
\end{theorem}

\begin{proof}

Note first that as $\check\Delta_A$ is self-adjoint, the orthogonal complement of ${\rm Im}(\check\Delta_A)$ is its kernel. Hence
\begin{equation*}
\check\Lambda^*(Y,E)={\rm Im}(\check\Delta_A)\oplus{\rm ker}(\check\Delta_A)
\end{equation*}
Moreover, it is easy to see that ${\rm Im}(\cd_A)$ and ${\rm Im}(\cd_A^\dagger)$ are orthogonal vector spaces, hence contained in ${\rm Im}(\check\Delta_A)$, and that they are both orthogonal to ${\rm ker}(\check\Delta_A)$. Indeed, consider e.g.
\begin{equation*}
\cd_A\beta=\alpha+\gamma\:,
\end{equation*}
where $\alpha\in{\rm Im}(\cd_A)$ and $\gamma\in{\rm ker}(\cd_A)$. It follows that
\begin{equation*}
(\gamma,\gamma)=(\gamma,\cd_A\beta-\alpha)=0\:,
\end{equation*}
and so $\gamma=0$. Similarly, one can show that ${\rm Im}(\cd_A^\dagger)\subseteq{\rm Im}(\check\Delta_A)$. We can then write a generic $\check\Delta_A\, \rho\in{\rm Im}(\check\Delta_A)$ as
\begin{equation*}
\check\Delta_A\, \rho=\cd_A\beta+\cd_A^\dagger\gamma+\kappa\:,
\end{equation*}
where $\kappa\in{\rm Im}(\check\Delta_A)$ is orthogonal to ${\rm Im}(\cd_A)$ and ${\rm Im}(\cd_A^\dagger)$. However, as ${\rm Im}(\check\Delta_A)$ is made up of sums of $\cd_A$-exact and $\cd_A^\dagger$-exact forms by construction of $\check\Delta_A$, it follows that $\kappa=0$. This concludes the proof.
\end{proof}

The Laplacian $\check\Delta_A$ is elliptic by construction (see Lemma \ref{lem:laplace} in appendix \ref{app:Elliptic}), and hence for compact $Y$ has a finite dimensional kernel. We refer to the kernel of $\check\Delta_A$ as harmonic forms and write
\begin{equation*}
\textrm{ker}\left(\check\Delta_A\right)=\mathcal{\check{H}}^*(Y,E)\:.
\end{equation*}
Moreover, it is easy to prove that $\mathcal{\check{H}}^*(Y,E)$ are in one to one correspondence with the cohomology classes of $H_{\check \dd_{A}}^*(Y,E)$ as usual. Indeed if $\alpha_1$ and $\alpha_2$ are harmonic representatives for the same cohomology class, then
\begin{equation*}
\alpha_1-\alpha_2=\cd_A\beta\:,
\end{equation*}
for some $\beta$. Applying $\cd_A^\dagger$ to this equation gives 
\begin{equation*}
\cd_A^\dagger\cd_A\beta=0\:,
\end{equation*}
which implies $\cd_A\beta=0$. Hence there is at most one harmonic representative per cohomology class. Moreover, if the class is to be non-trivial, by the Hodge-decomposition there must be at least one harmonic representative as well. Also, recall that by ellipticity of the complex, the cohomology groups $H_{\check \dd_{A}}^*(Y,E)$ are finite dimensional for compact $Y$.

%\newpage
\section{Infinitesimal moduli space of $G_2$ manifolds}
\label{sec:geommod}

We now discuss variations of $Y$ preserving the $G_2$ holonomy condition, a subject that has been discussed from different perspectives before. Firstly, Joyce has shown that, for compact $G_2$ manifolds, the infinitesimal moduli space maps to the space of harmonic three-forms, and thus has dimension $b^3$ \cite{joyce1996:1,joyce1996:2}.  Secondly, it has been shown by Dai {\it et.~al.} that this moduli space maps to the first $\cd$-cohomology group \cite{Dai2003}. This second result has also been found using a string theory analysis by de Boer {\it et.~al.} \cite{deBoer:2005pt}. In this section, we reproduce these results, using both the form and spinor description of the $G_2$ structure. 

Let $Y$ be a compact manifold with $G_2$ holonomy.  In this case
the three-form $\varphi$ is a harmonic three-form. Consider a one parameter family $Y_t$ of manifolds with a $G_2$ structure given by the associative three-form $\varphi_t$ with
$Y_0 = Y$ and $\varphi_0 =\varphi$. Below, we analyse the variations that preserve $G_2$ holonomy. For ease of presentation we relegate some of the details of the computation to \cite{delaOssa}, where variations of integrable $G_2$ structures will be discussed.

\subsection{Form perspective}\label{subsec:form}

Let us start by discussing the variation of $\psi$. This can be decomposed into $G_2$ representations as
\be
 \partial_t \psi = c_t\, \psi + \alpha_t\wedge\varphi + \gamma_t~,
  \label{eq:psihol}
 \ee
where $c_t$ is a function, $\alpha_t$ is a one-form, and $\gamma_t\in\Lambda^4_{27}$.  Equivalently, we may write the variation of $\psi$ (or any four form)  in terms of a one form $M_t$ with values in $TY$:
\begin{equation}
\partial_t\psi = \frac{1}{3!}\, M_t^a\wedge\psi_{bcda}\, \dd x^{bcd}
~,\qquad
M_t^a = M_{t\, b}{}^a\, \dd x^b~.\label{eq:varpsiM}
\end{equation}
We  can think of $M_t$ as a matrix, where its trace corresponds to forms in $\Lambda_1^4$ ({\it i.e.} $c_t$), its antisymmetric part ($\beta_{t\,ab}$) to $\Lambda_7^4$, and its  traceless symmetric part ($h_{t\,ab}$) to $\Lambda_{27}^4$. In particular, 
\begin{align}
c_t =  \frac{1}{7}\, \psi\lrcorner\partial_t\psi &= - \frac{4}{7}\,{\rm tr} M_t~,
\label{eq:ct}\\[5pt]
 \Delta_{t\, b}{}^a &= M_{t\, b}{}^a - \frac{1}{7}\, ({\rm tr} M_t)\, \delta_a{}^b~,
 \label{eq:Deltat}\\[3pt]
 \gamma_t &= \frac{1}{3!}\, h_t{}^a\wedge\psi_{bcda}\, \dd x^{bcd}\in \Lambda^4_{27}~,
 \qquad h_{t\, ab} = \Delta_{t\, (ab)}~,
 \label{eq:gammat}\\[3pt]
 \alpha_t &= \beta_t\lrcorner\varphi~, \qquad\beta_t = \frac{1}{2}\, \Delta_{t\, [ab]}\, \dd x^{ab}\in \Lambda^2_7(Y)~. 
\label{eq:alphat}
 \end{align}

The deformation of  $\varphi$ can be decomposed in an analogous manner. Moreover, using that $\psi = * \varphi$ one finds relations between the two variations, that give
\be
\partial_t \varphi = \hat c_t\, \varphi - \alpha_t\lrcorner\psi - \chi_t = - \frac{1}{2}\, M_t{}^a\wedge \varphi_{bca}\, \dd x^{bc}~,
 \label{eq:varphihol}
\ee
where $\hat c_t = 3\, c_t/4$ and $\gamma_t = *\chi_t$. Finally, using \eqref{eq:g2metric}, we may compute the variation of the $G_2$ metric:
\begin{equation}
\partial_t\, g_{\varphi\, ab} =  \frac{c_t}{2}\, g_{\varphi\, ab} - 2\,  h_{t\, ab}~,
\label{eq:varmetric}
\end{equation}
Note that the variation of the metric is only sensitive to the symmetric part of $\Delta^a$.

We now turn  to trivial deformations which correspond to diffeomorphisms. Again, we focus on $\psi$ (using the results above, we can compute the trivial variations of $\varphi$):
\be
{\cal L}_V \psi = \dd(v\lrcorner\psi) + v\lrcorner(\dd\psi)
= c_{triv}\, \psi+ \alpha_{triv}\wedge\varphi + \gamma_{triv}~,
\label{eq:liederpsi}
\ee
where ${\cal L}_V$ denotes a Lie derivative along vectors $V\in TY$, $v\in T^*Y$ is the one-form dual to $V$ using the metric, and we have included  the decomposition of the Lie derivatives in representations of $G_2$. The second term can be rewritten in terms of a two-form $\beta_{triv}\in \Lambda^2_7$ which is related to the one form $\alpha_{triv}$ by
\[ \beta_{triv} = \frac{1}{3}\, \alpha_{triv}\lrcorner\varphi~.\] We then have
\begin{theorem}\label{tm:diffeoG2}
On a $G_2$ manifold $Y$, deformations of the co-associative form $\psi$ due to diffeomorphisms of $Y$ are given by
\begin{equation}
{\cal L}_V\psi = - \frac{1}{3!}\, (\dth V^a)\wedge\psi_{bcda}\,\dd x^{bcd}~,\qquad V\in TY
\label{eq:trivvarpsi}
\end{equation}
where
\begin{equation}
 \dth V^a = \dd V^a + \theta_b{}^a\, V^b~,\qquad
 \theta_b{}^a = \Gamma_{bc}{}^a\, \dd x^c~,
 \label{eq:Dtheta}
 \end{equation}
is  a connection on $TY$, and $\Gamma_{bc}{}^a$ are the connection symbols of the Levi--Civita connection $\nabla$ compatible with the  $G_2$ structure on $Y$ determined by $\varphi$.  In fact, this is the connection $\dth$ defined in section \ref{subsec:thetaconnection}.

The correspondence with
\[ {\cal L}_V\psi = c_{triv}\, \psi + \alpha_{triv}\wedge\varphi + \gamma_{triv}~,\]
is given by
\begin{align}
c_{triv} &= \frac{4}{7}\, \nabla_a\, V^a 
= - \frac{4}{7}\, \dd^\dagger v~, 
\label{eq:ctriv}\\[3pt]
\beta_{triv} &= - \check \dd v ~,
\label{eq:betatriv}\\[3pt]
(h_{triv})_{ab} &= - \Big(\nabla_{(a} v_{b)} + \frac{1}{7}\, g_{\varphi\, ab}\, \dd^\dagger v\Big)~,
\label{eq:htriv}
\end{align} 
\end{theorem}

\begin{proof} This is proven by direct computation of the Lie derivatives. We relegate this proof to \cite{delaOssa}, where variations of integrable $G_2$ structures will be discussed.
\end{proof}
Note that if $Y$ is compact, by the Hodge decomposition of the function $c_t$ appearing in equation \eqref{eq:psihol}, equation \eqref{eq:ctriv} means that one can take $c_t$ to be a constant.  
Moreover, \eqref{eq:betatriv}, uses the $\cd$ differential operator defined in subsection \ref{subsec:checkD}.  
By the $\cd$-Hodge decomposition, we can write $\beta_t$ as
\[ \beta_t = \cd B_t + \cd^\dagger\lambda _t + \beta_t^{har}~,\]
for some  one form $B_t$, three form $\lambda_t$, and $\cd$-harmonic two form $\beta_t^{har}$. 
This means we can choose $\beta_t$ to be $\cd$-coclosed, which 
implies that $\alpha_t$ may be taken to be $\cd$-closed:
\begin{equation}
\cd\alpha = 0~.\label{eq:alphatrivG2}
\end{equation}
By the $\cd$-Hodge decomposition we can write $\alpha_t$ as
\[ \alpha_t = \cd A_t + \alpha_t^{har} = \dd A + \alpha_t^{har}~,\]
for some function $A_t$, and $\cd$-harmonic one form $\alpha_t^{har}$. Note however that there are no $\cd$-harmonic one forms on a compact manifold with $G_2$ holonomy \cite{fernandez1998dolbeault}, therefore $\alpha_t$ can be chosen to be $\dd$-exact
\[ \alpha_t = \dd A_t~.\]

We now require that the variations preserve the $G_2$ holonomy, that is, 
\begin{equation}
 \dd\partial_t \psi = 0~,\qquad \dd\partial_t \varphi = 0~.\label{eq:G2holintegral}
 \end{equation}
The first equation, together with equation \eqref{eq:psihol} gives
\[ \dd\gamma_t = 0 \iff \dd^\dagger\chi_t = 0~.\]
The second, together with \eqref{eq:varphihol}, gives
\[ \dd(\chi_t + \alpha_t\lrcorner\psi) = 0~.\]
However
\[\alpha_t\lrcorner\psi = (\dd A)\lrcorner\psi = - *((\dd A)\wedge\varphi) = -*\dd (A\,\varphi) = 
- \dd^\dagger(A\, \psi)~,\]
which implies
\[ \dd(\chi_t - \dd^\dagger(A\, \psi)) = 0~.
\]
We conclude then that the three form
\[  \chi_t +\alpha_t\lrcorner\psi  = \chi_t - \dd^\dagger(A\, \psi) ~,\]
is harmonic, and therefore the infinitesimal moduli space of manifolds with $G_2$ holonomy has dimension $b_3$, including the scale factor $c_t$.

We would like to compare this result with Joyce's proof \cite{joyce1996:1,joyce1996:2} that the dimension of the infinitesimal moduli space of manifolds with $G_2$ holonomy has dimension $b_3$.  Without entering into the details of the proof, Joyce finds the dimension of the moduli space by imposing conditions \eqref{eq:G2holintegral} together with
\begin{equation}
\pi_7(\dd^\dagger\partial_t\varphi) = 0~.\label{eq:Joyce}
\end{equation}
This constraint comes from requiring that the variations $\partial_t\varphi$ are orthogonal to the trivial deformations given by ${\cal L}_V\varphi$
\[(\partial_t\varphi, {\cal L}_V\varphi) = 0~,\qquad \forall~ V\in\Gamma(TY).\]
 In fact,
\[ (\partial_t\varphi, {\cal L}_V\varphi) = (\partial_t\varphi, \dd(v\lrcorner\varphi)) =
(\dd^\dagger(\partial_t\varphi), v\lrcorner\varphi)~,\]
which vanishes for all $V\in\Gamma(TY)$ if and only if \eqref{eq:Joyce} is satisfied, or equivalently, when
$\dd^\dagger(\partial_t\varphi)\in \Lambda^2_{14}$.  Now,
\[ \dd^\dagger(\partial_t\varphi) = - \dd^\dagger(\chi_t + \alpha_t\lrcorner\psi) = - \dd^\dagger(\alpha_t\lrcorner\psi)~,\]
as $\chi_t$ is co-closed. Taking the Hodge-dual of the constraint \eqref{eq:Joyce} we find
\begin{align*}
0&= * \big(\dd^\dagger(\partial_t\varphi)\wedge\psi\big)
=  *\big(\psi\wedge*\dd *(\alpha_t\lrcorner\psi)\big) 
= - *\big(\psi\wedge*\dd (\alpha_t\wedge\varphi)\big)
\\
&= - \psi\lrcorner(\dd\alpha_t\wedge\varphi) = \dd\alpha_t\lrcorner\varphi = \check \dd \alpha_t,
\end{align*}
which is the same as \eqref{eq:alphatrivG2}.

Finally, we would like to discuss the map between $\Delta$ and $\tilde\gamma$, in particular we would like to describe the moduli space of compact manifolds with $G_2$ holonomy in terms of $\Delta$. We begin with the moduli equations which for this case are
\begin{align}
\dth\Delta_t^a\wedge\psi_{bcda}\, \dd x^{bcd} &= 0~, \label{eq:g2holone} \\
\dth\Delta_t^a\wedge\varphi_{bca}\, \dd x^{bc} &= 0~.\label{eq:g2holtwo}
\end{align}
The second  equation is equivalent to 
\begin{align}
((\cdth\Delta_t{}^a)\lrcorner\varphi)_a &= 0
~,\label{eq:g2holtau0}
\\[5pt]  
(\pi_{14}(\dth \Delta_t^a))_{ba} &= 0~,
 \label{eq:g2holtau1}
 \\[5pt]
 (\dth\Delta^c)_{d(a}\varphi_{b)c}{}^d - g_{\varphi\, c(a}\, (\cdth\Delta^c\lrcorner\varphi)_{b)}
&=0~,
\label{eq:g2holtau3}
\end{align}
Note that 
equation \eqref{eq:g2holtau0} is just the trace of equation \eqref{eq:g2holtau3}.
Equation \eqref{eq:g2holone} can be better understood by contracting with $\varphi$ (the contraction with $\psi$ just gives back equation \eqref{eq:g2holtau1}). We find
\begin{equation}
2\, (\cdth\Delta_t{}^a)_{d[b}\, \varphi_{c]a}{}^d 
= - (\cdth\Delta_t{}^a)_{ad}\, \varphi_{bc}{}^d~. 
\label{eq:g2holmod}
\end{equation}

Then, applying equation \eqref{eq:id2rep7} to $\cdth\Delta^e$,
and contracting indices,  we find 
\[ (\cdth\Delta^a)_{da}\, \varphi_{bc}{}^d = 
(\cdth\Delta^a)_{d[b}\, \varphi_{c]a}{}^d  + g_{\varphi\, a[b}\, ((\cdth\Delta^a)\lrcorner\varphi)_{c]}~.\]
With this identity  at hand, we can write the equation for moduli \eqref{eq:g2holmod} as
\begin{equation} 
(\cdth\Delta_t{}^a)_{d[b}\, \varphi_{c]a}{}^d 
= g_{\varphi\, a[b}\, (\cdth\Delta^a\lrcorner\varphi)_{c]}~. 
\label{eq:g2holmod2}
\end{equation}
Adding up this equation and equation \eqref{eq:g2holtau3} we find
\begin{equation}
(\cdth\Delta_t{}^a)_{db}\, \varphi_{ca}{}^d 
+ \big(\pi_{14}(\dth\Delta_t{}^a)\big)_{d(b}\, \varphi_{c)a}{}^d 
= g_{\varphi\, ab}\, (\cdth\Delta^a\lrcorner\varphi)_c~.
\label{eq:g2holmod3}
\end{equation}
Using identity \eqref{eq:beta14id} in the second term
\begin{align*}
 \big(\pi_{14}(\dth\Delta_t{}^a)\big)_{d(b}\, \varphi_{c)a}{}^d 
&= \frac{1}{2}\, \Big(
2\, \pi_{14}\big((\dth\Delta_t{}^a)\big)_{db}\, \varphi_{ca}{}^d
-\pi_{14}\big((\dth\Delta_t{}^a)\big)_{da}\, \varphi_{cb}{}^d\Big)
\\[5pt]
&= \pi_{14}\big((\dth\Delta_t{}^a)\big)_{db}\, \varphi_{ca}{}^d~,
\end{align*}
where we have used equation \eqref{eq:g2holtau1}.  Hence equation \eqref{eq:g2holmod3}
becomes
\begin{equation}
(\dth\Delta_t{}^a)_{db}\, \varphi_{ca}{}^d 
= g_{\varphi\, ab}\, (\cdth\Delta^a\lrcorner\varphi)_c~.
\label{eq:g2holmod4}
\end{equation}

The derivative $\dth$ acts on $\Delta_t^a$ as the Levi-Civita connection when $Y$ has $G_2$ holonomy
\[ \dth \Delta_t^a =  \dd \Delta_t^a + \theta_b{}^a\wedge \Delta_t^b = \nabla_b\, \Delta_{t\,c}{}^a\,\dd x^{bc}~,\]
where $\nabla$ is the Levi-Civita connetion. 
Then
\[ \dth \Delta_t^a\lrcorner\varphi = 
\varphi^{bc}{}_d\, \nabla_b\Delta_{t\, c}{}^a\, \dd x^d~,\]
and equation \eqref{eq:g2holmod4} is equivalent to
\begin{equation}
\nabla_c\, h_{t\, da} \, \varphi^{cd}{}_b 
= \nabla_a(\beta_t\lrcorner\varphi)_b~.
\label{eq:g2holmod5}
\end{equation}
Taking the trace and using \eqref{eq:onephiphi} we find that
\[ 0 = \dd^\dagger(\beta_t\lrcorner\varphi) = \dd^\dagger \alpha~.\]
However, recall that by using diffeomorphisms we may choose $\alpha_t$ to be closed. It then follows that $\alpha_t$ is an harmonic one-form, and then has to vanish on compact manifolds with $G_2$ holonomy. We conclude that $\alpha_t$ and hence $\beta_t$ vanish, and so \eqref{eq:g2holmod5} implies that
\be
({\dth\Delta_t{}^a})\lrcorner \varphi  = 
\nabla_c\, h_{t\, d}{}^a \, \varphi^{cd}{}_b \dd x^b = 0~,
\ee
where we have used that $\beta_t=0$. Using Theorem \ref{tm:diffeoG2}, which states that  diffeomorphisms correspond to changing $\Delta^a$ by $\cd_{\theta}$-exact forms, we see that $\Delta^a$ remains $\cd_\theta$-closed under diffeomorphisms. We can then conclude that the infinitesimal moduli space of compact $G_2$ manifolds maps to the canonical $G_2$ cohomology group $H_{\check \dd_{\theta}}^1(Y,TY)$.

\subsection{Spinor perspective}
\label{eq:spinor}
We now derive again the results obtained in previous section from another perspective.  As the $G_2$ holonomy on the manifold $Y$ is determined by a well defined nowhere vanishing spinor $\eta$  which is covariantly constant, we study in this section the moduli of $Y$ by deforming the spinor and the $G_2$ holonomy condition. 

Let us first recall the definition of the fundamental three-form $\varphi$  and four form $\psi$ in terms of the Majorana spinor $\eta$,
\begin{align}
\varphi_{abc}&=-i\,\eta^\dagger\gamma_{abc}\eta\:,\label{eq:defvarphi2}\\
\psi_{abcd} &= -\eta^\dagger\gamma_{abcd}\eta\:.\label{eq:defpsi2}
\end{align}
The gamma-matrices satisfy the usual Clifford algebra
\begin{equation}
\{\gamma^\alpha,\gamma^\beta\}=2{\delta}^{\alpha\beta}\:.
\end{equation}
where $\gamma_a={e_{a}}^\alpha\gamma_\alpha$, and ${e_{a}}^\alpha$ denote the vielbein corresponding to the metric
\begin{equation}
 g_{ab} = e_a{}^\alpha\, e_b{}^\beta\, \delta_{\alpha\beta}~.\label{eq:metric}
 \end{equation}
We use labels $\{\alpha,\beta,..\}$ to denote tangent space flat indices.  We take the $\gamma$ matrices to be hermitian and imaginary. 
We will need below some $\gamma$ matrix identities which can be found in e.g. \cite{CANDELAS1984415}.
The $G_2$ holonomy condition on $Y$ can be expressed in terms of the spinor $\eta$ by the fact that it is covariantly constant with respect to the Levi-Civita connection
\begin{equation}
\label{eq:Susy1}
\nabla_a\eta_i=\p_a\eta_i+\frac{1}{4}\Omega_{a\,\alpha\beta}(\gamma^{\alpha\beta})_i{}^j\,\eta_j=0\:,
\end{equation}
where $\{i, j, \ldots\}$ are spinor indices.  
Here $\Omega_{a\,\alpha\beta}$ is the spin connection defined by $\nabla_a \, e_b{}^\alpha = 0$, that is
\begin{equation}
 \Omega_{a\,\alpha\beta}= - e^b{}_\beta\, (\partial_a\, e_b{}_\alpha - \Gamma_{ab}{}^c\, e_c{}_\alpha)~.
 \label{eq:SpinO}
 \end{equation}
 Note that the $\gamma$ matrices are covariantly constant\footnote{Indeed, the $\gamma$ matrices with flat tangent space indices are covariantly constant with respect to {\it any connection}.} .  In, fact
 \begin{equation*}
 \nabla_a\,(\gamma^b)= \partial_a\, (\gamma^b)+ \Gamma_{ac}{}^b\, (\gamma^c)
 - \frac{1}{4}\, \Omega_{a\alpha\beta} \, e_\gamma{}^b\, [\gamma^\gamma, \gamma^{\alpha\beta}]~,
 \end{equation*}
 and therefore
 \begin{equation*}
 \nabla_a\,\gamma^b = \left(\partial_a e_\alpha{}^b + \Gamma_{ac}{}^b\, e_\alpha{}^c 
 + \, \Omega_{a\alpha}{}^\beta\, e_\beta{}^b \right)\, \gamma^\alpha = (\nabla_a\, e_\alpha{}^b)\, \gamma^\alpha = 0~,
 \end{equation*}
 where we have used the $\gamma$ matrix identity
 \be \label{eq:gammaid1}
  [\gamma^\gamma, \gamma^{\alpha\beta}] =  4\, \delta^{\gamma\,[\alpha}\,\gamma^{\beta]}~.\ee

The moduli problem is discussed in this section in terms of those variations of $\eta$ and the vielbein $e_a{}^\alpha$ which preserve the $G_2$ holonomy condition
\eqref{eq:Susy1}. On manifold with a $G_2$ structure, a general variation of $\eta$ is given by
\begin{equation}
\label{eq:spinordef}
{\partial_t}\eta=d_t\,\eta+i\,{b_t}_a\gamma^a\eta\:,
\end{equation}
where $d_t$ is a real function and ${b_t}$ a real one form.  Any other terms would be of the form $\gamma^{ab}\eta$ or $\gamma^{abc}\eta$, however one can use the identities in equation (3.8) in \cite{Kaste:2003zd} to show that this is in fact the general form of an eight dimensional Majorana spinor on a manifold with a $G_2$ structure. A trivial deformation of $\eta$ corresponding to a diffeomorphism of $Y$ is given by the Lie derivative of $\eta$ along a vector $V$
\[
 {\cal L}_V\, \eta = V^a\, \nabla_a\eta + \frac{1}{4}\, \nabla_{[a}\, V_{b]} \, \gamma^{ab}\eta~.
 \]
Using the $G_2$ holonomy condition $\nabla_a\eta=0$ we obtain 
\begin{equation}
 {\cal L}_V\, \eta =  \frac{1}{4}\, \nabla_{[a}\, V_{b]} \, \gamma^{ab}\eta
 =  i\, \frac{1}{4}\, (\nabla_a V_b)\, \varphi^{ab}{}_c\, \gamma^c\, \eta ~,
 \end{equation}
 where in the last equality we have used  the identity (see \cite{Kaste:2003zd})
 \begin{equation}
  \gamma^{ab}\,\eta = i\, \varphi^{ab}{}_c \, \gamma^c\, \eta~.
  \label{eq:gammaid}
  \end{equation}
 Therefore trivial deformations of $\eta$ correspond to one forms $b_t$ in the variations of $\eta$ of the form
 \[ b_{triv} = \frac{1}{4}\, \dd v\lrcorner\varphi~.\]

The computation of the deformations of the $G_2$ holonomy condition \eqref{eq:Susy1} requires that we first compute  the variations of the Christoffel connection, the spin connection and the vielbein\footnote{These quantities can be found in the literature (see for example \cite{forger2004currents}), however we briefly sketch here the computations in order to make this section self contained.}.  The variations of the Christoffel connection are easily computed in terms of the variations of the metric
\begin{equation}
\partial_t\, \Gamma_{ab}{}^c = \frac{1}{2}\, g^{cd}\big( \nabla_a\, \partial_t g_{bc}
+ \nabla_b\, \partial_t g_{ac}  - \nabla_d\, \partial_t g_{ab}\big)~.\label{eq:varLC}
\end{equation}
The variations of the vielbein can be obtained from equation \eqref{eq:metric}
\begin{equation*}
\partial_t g_{ab} = 2 \big(\partial_t e_{(a}{}^\alpha\big)\, e_{b)}{}_\alpha
= 2\left( \big(\partial_t e_a{}^\alpha\big)\, e_b{}_\alpha - \Lambda_{t\, ab}\right)
~,
\end{equation*}
where we have defined 
\begin{equation}
\Lambda_{t\, ab} = \big(\partial_t e_{[a}{}^\alpha\big)\, e_{b]}{}_\alpha~.\label{eq:Lambda}
\end{equation}
Hence
\begin{equation}
\partial_t e_a{}^\alpha = e^{b\alpha}\,\left( \frac{1}{2}\,  \partial_t g_{ab} + \Lambda_{t\, ab}\right)~.
\label{eq:var7bein}
\end{equation}
Note that $\Lambda_t$ corresponds to deformations of the vielbein which do not change the metric. For the inverse of the vielbein we find
\begin{equation}
\partial_t e^a{}_\alpha = g^{ab}\, e^c{}_\alpha\, \left( - \frac{1}{2}\,  \partial_t g_{bc} + \Lambda_{t\, bc}\right)~.
\label{eq:var7beinInv}
\end{equation}
The variations of the spin connection are computed using equations \eqref{eq:varLC},  \eqref{eq:var7bein} and \eqref{eq:var7beinInv} and is left as an exercise for the reader
\begin{align}
\partial_t\Omega_{a\alpha\beta} &= e^b{}_\alpha\, e^c{}_\beta\big( \nabla_a\Lambda_{t\, bc} - \nabla_{[b}\, \partial_t g_{c]a}\big)~.
\label{eq:varspinO}
\end{align}

Next, we consider the deformations of the $G_2$ holonomy condition \eqref{eq:Susy1}.
Varying equation \eqref{eq:Susy1}, and using equations \eqref{eq:varspinO} and\eqref{eq:gammaid},  we have
\begin{align*}
\partial_t\nabla_a\eta&=\p_a\partial_t\eta+\frac{1}{4}\Omega_{a\,\alpha\beta}\gamma^{\alpha\beta}\partial_t\eta
 +\frac{1}{4}(\partial_t\Omega_{a\,\alpha\beta})\gamma^{\alpha\beta}\eta
= 
\nabla_a\partial_t\eta
 +\frac{1}{4}(\partial_t\Omega_{a\,\alpha\beta})\gamma^{\alpha\beta}\eta
 \\[5pt]
 &= \nabla_a\partial_t\eta
 +\frac{i}{4}\, \big( \nabla_a\Lambda_{t\, bc} - \nabla_{[b}\, \partial_t g_{c]a}\big)\,\varphi^{bc}{}_d\, \gamma^{d}\,\eta 
 \\[5pt]
 &= \nabla_a\left(\partial_t\eta + \frac{i}{2}\, (\Lambda_t\lrcorner\varphi)_b\, \gamma^b\eta\right)
 - \frac{i}{4}\, \nabla_{[b}\, \partial_t g_{c]a}\,\varphi^{bc}{}_d\, \gamma^{d}\,\eta
 \\[5pt]
 &= (\partial_a\, d_t)\, \eta + i\, \nabla_a\left( b_{t\, c} + \frac{1}{2}\, (\Lambda_t\lrcorner\varphi)_c\, \right)\, \gamma^{c}\,\eta
 - \frac{i}{4}\, \nabla_{[b}\, \partial_t g_{c]a}\,\varphi^{bc}{}_d\, \gamma^{d}\,\eta
  = 0
 \:,
\end{align*}
where in the last equality we have used equation \eqref{eq:spinordef} and the fact that the
$\gamma$ matrices are covariantly constant.  Therefore 
\[ \partial_a d_t = 0~,\] 
so $d_t$ is a constant, and 
\begin{equation}
 \nabla_a\left( b_{t\, d} + \frac{1}{2}\, (\Lambda_t\lrcorner\varphi)_d\, \right)
 = \frac{1}{4}\, \nabla_{[b}\, \partial_t g_{c]a}\,\varphi^{bc}{}_d~.
 \label{eq:g2holmod6}
\end{equation}
This equation is precisely  equation \eqref{eq:g2holmod5} as we discuss below.  

To compare  this analysis with our previous discussion in section \ref{subsec:form}, and in particular, to see how the moduli of the spinor and the vielbein are related to the moduli of $\varphi$, we consider an infinitesimal variation of $\varphi$ in equation \eqref{eq:defvarphi2}
\begin{align*}
\partial_t\varphi_{abc} &= -i\, \big(
\p_t\eta^\dagger \, \gamma_{abc}\, \eta + \eta^\dagger\, \gamma_{abc} \,\p_t\eta
+ 3\, (\p_t e_{[a}{}^\alpha)\, e^d{}_\alpha\eta^\dagger\, \gamma_{bc]d}\,\eta\big)
\\
&=
2\,d_t\,\varphi_{abc}
- b_{t\, d}\, \eta^\dagger\, [\gamma^d, \gamma_{abc}]\, \eta
 + 3\,e^d{}_\alpha\, ({\partial_t}{e_{[a}}^\alpha)\varphi_{bc]d}~.
\end{align*}
Using the $\gamma$ matrix identity
\[ [\gamma_{\alpha\beta\gamma}, \gamma_\delta] = 2\, \gamma_{\alpha\beta\gamma\delta}~,\]
and equation \eqref{eq:defpsi2}, we find
\begin{equation}
\label{eq:varphidef}
{\partial_t}\varphi=2\,d_t\,\varphi + 2\,{b_t}\lrcorner\psi 
+ \frac{1}{2}\, \,e^d{}_\alpha\, ({\partial_t}{e_{[a}}^\alpha)\varphi_{bc]d}\, \dd x^{abc}
\:.
\end{equation}
Comparing \eqref{eq:varphidef} with \eqref{eq:varphihol}
\begin{equation}
M_a{}^\alpha= -\frac{2}{3}\, d_t\,{e_a}^\alpha - \frac{2}{3}\,(b_t\lrcorner\varphi)_{ab}\, e^{b\alpha} - {\partial_t}{e_a}^\alpha\:.
\end{equation}
Moreover,
\begin{align*}
c_t &= \frac{8}{3}\, d_t + \frac{4}{7}\, (\p_t e_a{}^\alpha)\, e^a{}_\alpha~,
\\[5pt]
M_{(ab)} &= -\frac{2}{3}\, d_t\, g_{ab} - \frac{1}{2}\, \p_t g_{ab}
= - \left(\frac{4}{3}\, d_t + \frac{1}{7}\,(\p_t e_a{}^\alpha)\, e^a{}_\alpha\right)\, g_{ab} + h_{t\, ab}
~,
\\[5pt]
\beta_t\lrcorner\varphi&= -2\, b_t - \Lambda_t\lrcorner\varphi~,
\end{align*}
where $\beta_t$ is defined in equation \eqref{eq:alphat}.  The equation for moduli \eqref{eq:g2holmod6} then becomes
\begin{equation}
 \nabla_a  (\beta_t\lrcorner\varphi)_d)
 = -\frac{1}{2}\, \nabla_{[b}\, \partial_t g_{c]a}\,\varphi^{bc}{}_d~,
 \label{eq:g2holmod7}
\end{equation}
and hence, we obtain  the same conclusions as in subsection \ref{subsec:form}, as it should be. It is worth noting that the parameters $b_t$ and
$\Lambda_t$ do not contribute to deformations of the metric.  As we can shoose $\beta_t=0$ using diffeomorphisms, these two are related by
\[ 2\, b_t =- \Lambda_t\lrcorner\varphi~.\]
Moreover $\pi_{14}(\Lambda_t)$ corresponds to deformations of the vielbein which
do not change the $G_2$ structure.

%\newpage

\section{Infinitesimal moduli space of $G_2$ instanton bundles}
\label{sec:bundmod}

Consider a one parameter family of pairs $(Y_t, V_t)$ with $(Y_0, V_0)=(Y, V)$, where the curvature $F$ on the bundle $V$ satisfies the instanton equation $F\wedge\psi = 0$.  We want to study simultaneous deformations of the $G_2$ structure on $Y$ together with those of the bundle which preserve both the $G_2$ holonomy of $Y$ and the instanton equation. To achieve this, we deform the system to first order which gives the infinitesimal moduli space. We will then discuss how this result relates to an extension bundle, and finally give a few remarks on higher order obstructions.

\subsection{Form perspective}
\label{sec:bundmodform}

We start by varying the instanton equation: 
\[ 0 = \partial_t(F\wedge\psi) = \partial_t F\wedge\psi + F \wedge\partial_t\psi~.\]
Note that in the first term, the wedge product of $\partial_t F$ with $\psi$ picks out the part of  $\partial_t F$ which is in $\Lambda^2_7$. 
Noting that 
\[ \partial_t F = \dd_A\partial_t A~,\]
and contracting with $\psi$ we obtain
\begin{equation}\label{eq:InstMod}
 \pi_7(\dd_A\partial_t A)= \check \dd_A \partial_t A 
= - \frac{1}{3}\, \psi\lrcorner (F \wedge\partial_t\psi)~,
\end{equation}
where we have used equation \eqref{eq:twophiphi}.

Keeping the $G_2$ structure fixed ($\partial_t\psi = 0$) on the base manifold gives the equation for the bundle moduli, that is, let $t$ be a bundle parameter, then 
\begin{equation} 
\check \dd_A \partial_t A = 0~,\label{eq:Bmod}
\end{equation}
Moreover, it is clear that $\check \dd_A$-exact one-forms correspond to gauge transformations, so
the bundle moduli are in correspondence with the cohomology group $H_{\check \dd_A}^1(Y, \rm{End}(V))$.\footnote{Recall that under an infinitesimal gauge transformation $\epsilon\in\Omega^0(\End(V))$, the connection transforms as
$A\rightarrow A+\dd_A\epsilon=A+\cd_A\epsilon$.}

Suppose now that $t$ is a deformation of the $G_2$ structure.  Then equation \eqref{eq:InstMod} is a constraint on the geometric moduli $\Delta_t$, that is required if the deformed  bundle connection shall be an instanton. Recall that we may decompose the variations of $\psi$ as 
\[ \partial_t\psi = c_t\, \psi + \tilde\gamma_t~,\]
where $c_t$ is a constant, and $\tilde\gamma_t\in \Lambda^4_7 + \Lambda^4_{27}$ is related to the traceless matrix $\Delta_{t\, ab}$:
\[ \tilde\gamma_t =  \frac{1}{3!} \, \Delta_{t\, a}{}^e\, \psi_{bcde}\,\dd x^{abcd}~.\]
 
To understand the right hand side of equation \eqref{eq:InstMod} we define the map
\begin{align*}
{\cal F} :\quad  \Lambda^p(Y, TY)\qquad &\longrightarrow\qquad \Lambda^{p+1}(Y, {\rm End}(V))\\
 \Delta\quad\qquad&~\mapsto \qquad 
{\cal F}(\Delta) = - F_{ab}\,\dd x^b\wedge\Delta^a~.
 \end{align*}
 We also define the map 
 \begin{equation*}
\check{\cal F} :\quad  \Lambda_{\bf r}^p(Y, TY)\qquad \longrightarrow\qquad \Lambda_{\bf r'}^{p+1}(Y, {\rm End}(V))~, 
\end{equation*}
where $\Lambda_{\bf r}^p(Y, {\rm End}(V))\subseteq \Lambda^p(Y, {\rm End}(V))$, 
$\Lambda_{\bf r'}^{p+1}(Y, {\rm End}(V))\subseteq \Lambda^{p+1}(Y, {\rm End}(V))$,  and ${\bf r}$ and ${\bf r'}$ are appropriate irreducible $G_2$ representations as follows:
\begin{align*}
 \check{\cal F}(\Delta) &= {\cal F}(\Delta) = -  F_{ab}\,\dd x^{b}\, \Delta^a~, &{\rm for}\quad\Delta\in \Lambda^0(TY)~,
 \\
 \check{\cal F}(\Delta) &= \pi_7({\cal F}(\Delta)) =  - \pi_7(F_{ab}\,\dd x^{b}\wedge \Delta^a)~, &{\rm for}\quad\Delta\in \Lambda^1(TY)~,
 \\
  \check{\cal F}(\Delta) &= \pi_1({\cal F}(\Delta)) = -  \pi_1(F_{ab}\,\dd x^{b}\wedge \Delta^a)~, &{\rm for}\quad\Delta\in \Lambda_7^2(TY)~.
\end{align*}
Note that the projections that define $\check{\cal F}$ are completely analogous to those that define the derivatives $\check \dd_A$.  It will become clear why we need this map shortly. 

 \begin{proposition}
 The equation for the moduli of instantons, \eqref{eq:InstMod}, is equivalent to
 \begin{equation}
  \check \dd_A(\partial_t A) = - \check{\cal F}(\Delta_t)~,
  \label{eq:Atiyah}
  \end{equation}
 where $\Delta_t\in \Lambda^1(Y, TY)$.
 \end{proposition}
 \begin{proof}
 The proof is a straightforward computation.
 \begin{align*}
 \psi\lrcorner(F\wedge\tilde\gamma) 
 &= \frac{6\cdot 5}{4\cdot 4!}\, \psi^{c_1c_2c_3c_4}\,F_{[c_1c_2}\, \tilde\gamma_{c_3c_4ab]}\,\dd x^{ab}
 \\
 &= \frac{1}{2\cdot 4!}\, \psi^{c_1c_2c_3c_4}\, ( 6\, F_{c_1c_2}\, \tilde\gamma_{c_3c_4ab}
 + 8\, F_{c_1a}\, \tilde\gamma_{bc_2c_3c_4} + F_{ab}\, \tilde\gamma_{c_1c_2c_3c_4})\,\dd x^{ab}
 \\[5pt]
 &= \frac{1}{4!}\, \psi^{c_1c_2c_3c_4}\, ( 3\, F_{c_1c_2}\, \tilde\gamma_{c_3c_4ab}
 + 4\, F_{c_1a}\, \tilde\gamma_{bc_2c_3c_4})\,\dd x^{ab}
 + (\psi\lrcorner\tilde\gamma)\, F~.
 \end{align*}
 The last term vanishes because $\tilde\gamma\in \Lambda^4_7 + \Lambda^4_{27}$, and using $F = - F\lrcorner\psi$ (as $F\in\Lambda^2_{14}$) in the first term we have 
 \begin{equation}
 \psi\lrcorner(F\wedge\tilde\gamma) 
 = \Big(- \frac{1}{4}\, F^{c_1c_2}\, \tilde\gamma_{c_1c_2ab}
 + \frac{1}{3!}\, \psi^{c_1c_2c_3c_4}\, F_{c_1a}\, \tilde\gamma_{bc_2c_3c_4}
 \Big)\dd x^{ab}~.\label{eq:preprop}
 \end{equation}
 By equations \eqref{eq:ct}, \eqref{eq:gammat} and \eqref{eq:alphat}, it is easy to check 
 that
 \[ \frac{1}{3!}\, \psi^{ac_1c_2c_3}\, \tilde\gamma_{bc_1c_2c_3}
 = - 2 \, \Delta_{t\, b}{}^a - \Delta_{t\,d}{}^c\, \psi^{ad}{}_{bc}~.\]
 Therefore, \eqref{eq:preprop} becomes
 \begin{align*}
 \psi\lrcorner(F\wedge\tilde\gamma) 
 &=\big(-  F^{c_1c_2}\, \Delta_{t\, [c_1}{}^d\, \psi_{c_2 ab]d} 
 + (
 - 2 \, \Delta_{t\, b}{}^c - \Delta_{t\,d}{}^e\, \psi^{cd}{}_{be})\, F_{ca}
  \big)\,\dd x^{ab}
  \\[3pt]
  &= - \frac{1}{2}\,\big( F^{c_1c_2}\, (\Delta_{t\, c_1}{}^d\, \psi_{c_2 abd}
 + \Delta_{t\, a}{}^d\, \psi_{bc_1c_2d})  
 +2\, (
 2 \, \Delta_{t\, b}{}^c + \Delta_{t\,d}{}^e\, \psi^{cd}{}_{be})\, F_{ca}
  \big)\,\dd x^{ab}
  \\[4pt]
   &= - \frac{1}{2}\, \big( 
   2\,   F_{ca}\, \Delta_{t\, b}{}^c
   -  F^{c_1c_2}\,  \psi_{c_1 abd}\, \Delta_{t\, c_2}{}^d
    +2\, F_{ca}\, \psi^{cd}{}_{be}\, \Delta_{t\,d}{}^e
 \big)\,\dd x^{ab}~,
  \end{align*}
  where in the last step we have used again $-F = F\lrcorner\psi$.
  We now use the identity (see equation \eqref{eq:betapsi} proven in  Appendix \ref{app:formulas})
 \[  F^{ae}\, \psi_{ebcd}= 3\, F_{e[b}\,\psi_{cd]}{}^{ae}~,\]
 in the second term, and we obtain
 \begin{align*}
 \psi\lrcorner(F\wedge\tilde\gamma) 
 &= - \frac{1}{2}\, \big( 
   2\,   F_{ca}\, \Delta_{t\, b}{}^c
   - 3\, F_{c[a}\, \psi_{be]}{}^{cd}\, \Delta_{t\, d}{}^e
   +2\, F_{ca}\, \psi_{be}{}^{cd}\, \Delta_{t\,d}{}^e
 \big)\,\dd x^{ab}
 \\[4pt]
 &=   - \frac{1}{2}\, \big( 
   2\,   F_{ca}\, \Delta_{t\, b}{}^c
  - F_{ce}\, \psi_{ab}{}^{cd}\, \Delta_{t\, d}{}^e
   \big)\,\dd x^{ab}
   \\[4pt]
 &= - \frac{1}{2}\, F_{cd}\, \Delta_{t\, e}{}^c\, 
 ( 2\, \delta_a^d\, \delta_b^e\, + \psi_{ab}{}^{de})\,\dd x^{ab}
   \\[4pt]
   &= \frac{1}{4}\, {\cal F}(\Delta_t)_{cd}\, 
 ( 2\, \delta_a^c\, \delta_b^d\, + \psi_{ab}{}^{cd})\,\dd x^{ab}
 = \frac{1}{4}\, {\cal F}(\Delta_t)_{cd}
 \, \varphi^{cd}{}_e\, \varphi_{ab}{}^e \,\dd x^{ab}
   \\[4pt]
   &= ({\cal F}(\Delta_t)\lrcorner\varphi)\lrcorner\varphi
   = 3\, \pi_7({\cal F}(\Delta_t))
 = 3\, \check{\cal F}(\Delta_t)~.
 \end{align*}
 This result, together with equation \eqref{eq:InstMod}, gives equation \eqref{eq:Atiyah}.
  
 \end{proof}

The map $\check{\cal F}$ is actually a map between cohomologies, moreover, it maps the metric moduli space of $G_2$ manifolds into the $\check \dd_A$-cohomology.  As we will see below this is a consequence of the Bianchi identity $\dd_A F = 0$.  We begin with a useful lemma. 

 \begin{lemma}
 The exterior covariant derivative $\dd_A$ of ${\cal F}$ is given by
 \begin{equation}
\dd_A({\cal F}(\Delta)) + {\cal F}(\dd\Delta) = - \dd_A (F_{ab}\, \dd x^b)\wedge\Delta^a ~.
\label{eq:cool}
\end{equation}
for any $p$-form $\Delta$ with values in $TY$.
Moreover, due to the Bianchi identity
\[\dd_A F = 0~,\]
 the right hand side of equation \eqref{eq:cool} becomes
\begin{equation}
\dd_A({\cal F}(\Delta)) + {\cal F}(\dd\Delta) 
= - (\partial_a F + A_a\, F - F\, A_a)\wedge\Delta^a~.\label{eq:cooltwo}
\end{equation} 
 \end{lemma}
 \begin{proof}
 \begin{align*} 
\dd_A({\cal F}(\Delta)) 
&= -  \dd_A(F_{ab}\, \dd x^b\wedge\Delta^a) \\
&= -  \dd (F_{ab}\, \dd x^b\wedge\Delta^a) -  A\wedge F_{ab}\, \dd x^b\wedge\Delta^a
+ F_{ab}\, \dd x^b\wedge\Delta^a\wedge A\\
& = - \dd_A (F_{ab}\, \dd x^b)\wedge\Delta^a - {\cal F}(\dd\Delta)~,
\end{align*}
where
\[{\cal F}(\dd\Delta) = - F_{ab}\, \dd x^b\wedge\dd\Delta^a~.\]
To obtain equation \eqref{eq:cooltwo}, we re-write the Bianchi identity 
\[ 0= \dd_A F = \dd F + A\wedge F - F \wedge A~,\]
 in a form that will prove very useful.  
\begin{align*}
0&= (\dd_A F)_{abc}\, \dd x^{bc} = 
3\big(\partial_{[a}F_{bc]} + A_{[a}\,F_{bc]} - F_{[ab}\, A_{c]}\big)\, \dd x^{bc}
\\
&= 2\,(\partial_a F + A_a\, F - F\, A_a) - 2 \, \dd_A(F_{ab}\, \dd x^b)~.
\end{align*}
Hence
\begin{equation} 
\dd_A(F_{ab}\, \dd x^b) = \partial_a F + A_a\, F - F\, A_a~.
\label{eq:BItwo}
\end{equation}
Using this equation into \eqref{eq:cool} we obtain \eqref{eq:cooltwo}.  Note that this proof makes it clear that this equation is covariant. In fact, one can make this explicit by writing it as
\[\dd_A({\cal F}(\Delta)) + {\cal F}(\dth\Delta) 
= - (\dd_A (F_{ab}\, \dd x^b) - \theta_a{}^b\wedge F_{bc}\, \dd x^c)\wedge\Delta^a~.\]
See theorem below for details. 

 \end{proof}

 \begin{theorem}
 Let $\Delta$ be a $p$-form with values in $TY$. The Bianchi identity $\dd_A F = 0$ implies that the map $\check{\cal F}$ satisfies 
 \begin{equation}
 \check{\cal  F}(\check \dth(\Delta)) + \check \dd_A(\check{\cal F}(\Delta))= 0~.\label{eq:idcohom}
 \end{equation}
This implies that forms $\Delta\in \Lambda^p(Y, TY)$ which are $\check \dth$-exact
 are mapped into $\check \dd_A$-exact forms in $\Lambda^{p+1}(Y, {\rm End}(V))$.
 Therefore, $\check{\cal F}$ maps the infinitesimal moduli space of $Y$, given by elements of $H^1_{\cd_\theta}(Y,TY)$, into elements of the cohomology
$H^2_{\check \dd_A}(Y, {\rm End}(V))$.\\
 
 \end{theorem}
 
 \begin{proof}
 To compute $\check \dd_A(\check {\cal F}(\Delta))$, we need to consider the map ${\cal }$ acting on $\Delta\in\Lambda^p(Y,TY)$ for each $p$. For $p=0$,
  \begin{equation}
 \check \dd_A(\check{\cal F}(\Delta)) = \check \dd_A({\cal F}(\Delta))
= \pi_7(\dd_A({\cal F}(\Delta))) = \frac{1}{3}\, \big(\dd_A({\cal F}(\Delta))\lrcorner\varphi\big)\lrcorner\varphi
~.\label{eq:pzero}
\end{equation}
For $p=1$, we have
\begin{equation*}
\check \dd_A(\check{\cal F}(\Delta)) = \check \dd_A(\pi_7({\cal F}(\Delta))) = 
\pi_1(\dd_A(\pi_7({\cal F}(\Delta))))= \check \dd_A({\cal F}(\Delta))~,
\end{equation*}
because, $\check \dd_A(\pi_{14}({\cal F}(\Delta))) = 0$ automatically.
In fact, by Lemma \ref{lem:dAontwo14}, for any $\beta\in \Lambda_{14}^2$, 
\[\check \dd_A(\beta) \in \Lambda_7^3\oplus\Lambda_{14}^3~.\]
Hence
\begin{equation}
\check \dd_A(\check{\cal F}(\Delta)) = 
\check \dd_A({\cal F}(\Delta)) 
= \pi_1(\dd_A({\cal F}(\Delta)))= 
\frac{1}{7}\, \big(\varphi\lrcorner \dd_A({\cal F}(\Delta))\big)\, \varphi~.
\label{eq:pone}
 \end{equation}
 Finally, when $p=2$, we have
 \[
 \check \dd_A(\check{\cal F}(\Delta)) = \check \dd_A(\pi_1({\cal F}(\Delta))) = 0~.
 \]
 
 Next,  we need to consider the projections onto $\Lambda_7^2$ and $\Lambda_1^3$, for $p = 0, 1$, respectively, as shown in equations \eqref{eq:pzero} and \eqref{eq:pone}.  Note that in both cases we need to compute the contraction of  $\d\dd_A({\cal F}(\Delta))$ with $\varphi$.  Recall equation \eqref{eq:cooltwo}
 \begin{equation*}
\dd_A({\cal F}(\Delta)) + {\cal F}(\dd\Delta) 
= - (\partial_a F + A_a\, F - F\, A_a)\wedge\Delta^a~.
\end{equation*}
Contracting with $\varphi$ and using $\varphi\lrcorner F = 0$, we have
 \begin{align}
\big(\dd_A({\cal F}(\Delta)) + {\cal F}(\dd\Delta)\big)\lrcorner\varphi &= 
- \big(\partial_a F\wedge \Delta^a\big)\lrcorner\varphi
= - * \big(\partial_a F\wedge \Delta^a\wedge\psi
\big)\nn\\
&=  * \big(F\wedge(\partial_a \psi)\wedge \Delta^a\big)
\nn\\
&= \frac{1}{3!}\, * \Big( F\wedge \psi_{cdeb}\, \dd x^{cde}\, \wedge \theta_a{}^b \wedge\Delta^a\Big)
~,\label{eq:precomp}
 \end{align}
 where in the last step we have used Lemma \ref{lem:idone}, and where
 \[ \theta_a{}^b = \Gamma_{ac}{}^b\, \dd x^c~,\]
 with $\Gamma$ being the connection symbols for a metric connection $\nabla$ which is compatible with the integrable $G_2$ structure $\varphi$ on $Y$.
 By Lemma \ref{lem:idfortwo14}, equation \eqref{eq:precomp} gives
 \begin{align*}
 \big(\dd_A({\cal F}(\Delta)) + {\cal F}(\dd\Delta)\big)\lrcorner\varphi 
  &= 
 * \Big( F_{bc}\,\dd x^c\, \wedge\psi \wedge \theta_a{}^b \wedge\Delta^a\Big)
 = - *\Big( {\cal F}(\theta_a{}^b \wedge\Delta^a)\wedge\psi\Big)
 \\
 &= - \big( {\cal F}(\theta_a{}^b \wedge\Delta^a)\big)\lrcorner\varphi~.
 \end{align*}
 Recall that we have defined the connection $\dth$ as
 \begin{equation*}
 \dth\Delta^a = \dd\Delta^a + \Gamma_{bc}{}^a\, \dd x^c\wedge\Delta^b 
= \dd \Delta^a + \theta_b{}^a\,\wedge \Delta^b~,
 \qquad \Delta\in \Lambda^p(Y, TY)~. \end{equation*}
 Therefore 
 \begin{equation}\label{eq:coolthree}
 \big(\dd_A({\cal F}(\Delta)) + {\cal F}(\dth\Delta)\big)\lrcorner\varphi = 0~.
 \end{equation}
 Returning to equations \eqref{eq:pzero} and \eqref{eq:pone} we find
 \begin{align*}
  \check \dd_A(\check{\cal F}(\Delta))  
&=  - \frac{1}{3}\, \big({\cal F}(\dth\Delta)\lrcorner\varphi\big)\lrcorner\varphi
= - \pi_7({\cal F}(\dth\Delta)) = - \check{\cal F}(\dth\Delta)~,
 \\[5pt]
 \check \dd_A(\check{\cal F}(\Delta)) &=  
- \frac{1}{7}\, \big({\cal F}(\dth\Delta)\lrcorner\varphi\big)\, \varphi
= - \pi_1({\cal F}(\dth\Delta)) = -  \check{\cal F}(\dth\Delta) 
= -  \check{\cal F}(\check \dth\Delta)~.
 \end{align*}
 where in the last step in the second equation we have used the fact that for any two form $\beta\in\Lambda_{14}^2(Y,TY)$,
 \[ {\cal F}(\beta)\wedge\psi = - F_{ab}\, \dd x^b\wedge\beta^a\wedge\psi = 0~,\]
implying that  ${\cal F}(\beta) \in \Lambda_7^3\oplus\Lambda_{14}^3$. Therefore
\[  \check{\cal F}(\dth\Delta) = 
\check{\cal F}\big(\check \dth\Delta + \pi_{14}(\dth\Delta)\big)
= \check{\cal F}(\check \dth\Delta)~.\]

 \end{proof}
 
\subsection{Spinor perspective}
\label{sec:bundmodspin}

It can be useful to understand the infinitesimal deformations of the instanton condition from more perspectives. Let us therefore describe the infinitesimal deformations in terms of the spinorial perspective of section \ref{eq:spinor}. 

Recall that the gaugino supersymmetry condition reads
\begin{equation}
\label{eq:instcondbund}
F_{ab}\gamma^{ab}\eta=0\:.
\end{equation}
On a manifold with a $G_2$-structure, this is equivalent to \eqref{eq:G2inst}. A generic variation of \eqref{eq:instcondbund} gives
\be \label{eq:gauginovar}
\partial_t(F_{ab}\gamma^{ab})\eta +
F_{ab}\gamma^{ab} \partial_t\eta=0
\ee
where in the second term we can use \eqref{eq:spinordef}:
\[
F_{ab}\gamma^{ab} \partial_t\eta = F_{ab}\gamma^{ab} (d_t\eta + i b_{t c} \gamma^c \eta) = i F_{ab} b_{t c} ([\gamma^{ab},\gamma^c] +\{\gamma^{ab},\gamma^c\})\eta =i F_{ab} b_{t c} (-4 g^{ca}\gamma^{b} +2 \gamma^{abc})\eta \; .
\]
The second equality is a consequence of the gaugino supersymmetry condition, and in the last step we use \eqref{eq:gammaid1} and the identity
\be
\{\gamma^{ab},\gamma^c\} = 2 \gamma^{abc}\, .
\ee
We then use the identities in equation (3.8) in \cite{Kaste:2003zd}, and that $F \in \Lambda^2_{14}(Y,\rm{End}(V))$, which implies that $F \lrcorner \varphi = 0$ and $-F = F\lrcorner\psi$: 
\be
F_{ab} \gamma^{abc} \eta = F_{ab} (i \varphi^{abc} + \psi^{abcd}\gamma_d)\eta = -2 F^c{}_{d} \gamma^d \eta
\; .
\ee
Consequently,
 \[
F_{ab}\gamma^{ab} \partial_t\eta = -8 i  F^a{}_{b} b_{t a} \gamma^b \eta = -\frac{4}{3} ((b_t \lrcorner F) \lrcorner \varphi)_{ab} \gamma^{ab}\eta \; ,
\]
where we have used \eqref{eq:gammaid}. 
The remaining terms in \eqref{eq:gauginovar} are given by
\[
\begin{split}
\big[\partial_t(F_{ab})\gamma^{ab} + F_{ab} e_c{}^{\alpha} \partial_t e^a{}_{\alpha} \gamma^{bc}\big]\eta &= 
\big[{\dd_A}_{[a}\partial_tA_{b]} +
F_{ca} e_b{}^{\alpha} \partial_t e^c{}_{\alpha}
\big]\gamma^{ab}\eta \\
&=
\big[{\dd_A}_{[a}\partial_tA_{b]} +
F_{ca} g^{cd} \big(-\frac{1}{2} \partial_t g_{bd} - \Lambda_{t \, bd}\big)
\big]\gamma^{ab}\eta \\
&=
\big[{\dd_A}_{[a}\partial_tA_{b]} +
\big(h_{t\,a}{}^c - \Lambda_{t \, a}{}^c\big)F_{bc} 
\big]\gamma^{ab}\eta
\end{split}
\]
where we have used \eqref{eq:var7beinInv} and \eqref{eq:instcondbund}. We conclude
\[
\partial_t(F_{ab}\gamma^{ab}\eta)=
\big[{\dd_A}_{[a}\partial_tA_{b]} +
\big(h_{t\,a}{}^c - \Lambda_{t \, a}{}^c\big)F_{bc} 
-\frac{4}{3} ((b_t \lrcorner F) \lrcorner \varphi)_{ab}
\big]\gamma^{ab}\eta
\:.
\]
We then note that
\[
\begin{split}
\big[\Lambda_{t \, a}{}^c F_{bc} +\frac{4}{3} ((b_t \lrcorner F) \lrcorner \varphi)_{ab} 
\big]\gamma^{ab}\eta &=
i \big[\Lambda_{t \, a}{}^c F_{bc} +\frac{4}{3} ((b_t \lrcorner F) \lrcorner \varphi)_{ab} 
\big]\varphi^{abd} \gamma_{d}\eta \\
&=
i \big[-(\Lambda_{t}\lrcorner\varphi)\lrcorner F +2 (b_t \lrcorner F)\big]_d \gamma^{d}\eta \\
&=
- i [(\beta_{t}\lrcorner\varphi)\lrcorner F]_d \gamma^{d}\eta
=
-\beta_{t \, a}{}^c F_{bc} 
\gamma^{ab} \eta \; .
\end{split}
\]
Here we have used $F \in \Lambda^2_{14}(Y,\rm{End}(V))$, which ensures that
\be
\label{eq:FcheckLambda}
\begin{split}
&{\cal F}(\Lambda_t) \lrcorner \varphi= 
\Lambda_{t \, a}{}^c F_{bc}\varphi^{ab}{}_d \dd x^d = 
\Lambda_{t}^{ac} (3 F_{b[c}\varphi_{da]}{}^b - F_{bd} \varphi_{ac}{}^b - F_{ba}\varphi_{cd}{}^b )\dd x^d \\
\implies 
&\check{\cal F}(\Lambda_t)  = {\cal F}(\Lambda_t) \lrcorner \varphi =
- (\Lambda_t \lrcorner \varphi) \lrcorner F = {\cal F}(\Lambda_t \lrcorner \varphi) \; ,
\end{split}
\ee
where the maps $\check{\cal F}, \cal F$ were introduced in section \ref{sec:bundmodform}.

Variations that preserve the gaugino supersymmetry equation must thus satisfy
\begin{equation}
\label{eq:Atiyah2}
\left({\dd_A}_{[a}\partial_tA_{c]}+{\Delta_{t\,a}}^d\,F_{dc}\right)\gamma^{ac}\eta=0\:,
\end{equation}
which is equivalent to \eqref{eq:Atiyah}, upon using \eqref{eq:gammaid}. We hence arrive at the same condition from the spinorial perspective as well. This was of course expected, as the two descriptions should be equivalent.

\subsection{The infinitesimal moduli space}
\label{sec:bundmod}
The constraint \eqref{eq:Atiyah} (or equivalently \eqref{eq:Atiyah2}) on the variations $\Delta_t\in {\cal TM}$ of the $G_2$-holonomy metric of $Y$, means that $\check{\cal F}(\Delta_t)$ must be
$\check \dd_A$-exact, that is
\[ \Delta_t\in {\rm ker}(\check{\cal F})\subseteq {\cal TM}~.\]
 Therefore, the tangent space of the moduli space of the combined deformations of $G_2$-holonomy metrics and bundle deformations is given by
\[ {\cal T}{\cal M} = H^1_{\check \dd_A}(Y, {\rm End}(V))\oplus {\rm ker}(\check{\cal F})~,\]
where elements in $H^1_{\check \dd_A}(Y, {\rm End}(V))$ correspond to bundle moduli. Recall however that the infinitesimal moduli space of $G_2$-holonomy metrics does not span the cohomology group $H^1_{\cd_\theta}(Y,TY))$. Let us take a closer look at this.

We first recall the isomorphism between the cohomology $H^1_{\cd_\theta}(Y,TY))$ and the harmonic one-forms,
\begin{equation*}
H^1_{\cd_\theta}(Y,TY))\cong{\cal\check H}^1(Y,TY)\:,
\end{equation*}
where the harmonic forms are in the kernel of the laplacian
\begin{equation*}
\check\Delta_\theta\::\;\;\;\Omega^*(Y,TY)\rightarrow\Omega^*(Y,TY)\:,
\end{equation*}
which is constructed using the Levi-Civita connection and the $G_2$-holonomy metric. It is easy to check that the diffeomorphism gauge which sets $\alpha_t=0$ in \eqref{eq:g2holmod5} ensuring the $\cd_\theta$-closure of $h_t$, also makes $h_t$ harmonic with respect to the Laplacian $\check\Delta_\theta$ as an element of $\Omega^1(Y,TY)$. Indeed, recall that in this gauge we have that the three-form
\begin{equation*}
\chi_t=\frac{1}{2}{h_{t\,a}}^d\varphi_{bcd}\,\dd x^{abc}
\end{equation*}
is harmonic. Applying $\dd^\dagger$ to $\chi_t$ we find
\begin{equation*}
\left(\nabla^a{h_{t\,a}}^d\right)\varphi_{bcd}+2\left(\nabla^a{h_{t\,[b}}^d\right)\varphi_{c]ad}=0\:.
\end{equation*}
The last term vanishes by the symmetry property of $h_{t\,ab}$ and the $\cd_\theta$-closure of $h_t$ as an element of $\Omega^1(Y,TY)$. From this it follows that 
\begin{equation*}
\nabla^a{h_{t\,a}}^d=0\:,
\end{equation*}
establishing the co-closure of $h_t$. Now, the harmonic forms further decompose as
\begin{equation*}
{\cal\check H}^1(Y,TY)={\cal\check S}^1(Y,TY)\oplus{\cal\check A}^1(Y,TY)\:,
\end{equation*}
where ${\cal\check S}^1(Y,TY)$ denote the {\it symmetric} elements of ${\cal\check H}^1(Y,TY)$, viewed as a $(7\times7)$-matrix. These are precisely the traceless symmetric deformations $h_t$, plus the singlet deformation corresponding to a re-scaling of $\varphi$. Together, these span all the non-trivial deformations of the three-form $\varphi$. We must therefore have
\begin{equation*}
{\cal\check S}^1(Y,TY)\cong H^3(Y)\:.
\end{equation*}
${\cal\check A}^1(Y,TY)$ denote the {\it anti-symmetric} elements of ${\cal\check H}^1(Y,TY)$, viewed as a $(7\times7)$-matrix. 

There is a further decomposition of ${\cal\check A}^1(Y,TY)$ into the $\mathbf 7$ and $\mathbf{14}$ representations
\begin{equation*}
{\cal\check A}^1(Y,TY)={\cal\check A}_{\mathbf 7}^1(Y,TY)\oplus {\cal\check A}_{\mathbf{14}}^1(Y,TY)\:.
\end{equation*}
However, on a compact manifold of $G_2$-holonomy it can be shown that 
\begin{equation*}
{\cal\check A}_{\mathbf 7}^1(Y,TY)=0\:.
\end{equation*}
The proof of this follows a similar procedure to the argument of section \ref{sec:geommod}, where it was shown that the one-form $\alpha_t$ could be set to zero by an appropriate diffeomorphism. Basically, one can use an element of ${\cal\check A}_{\mathbf 7}^1(Y,TY)$ to construct a $\cd$-harmonic one-form. This contradicts the fact that there are no such one forms on a compact $G_2$-holonomy manifold. The remaining $\mathbf{14}$-representation has an interpretation as $B$-field deformations \cite{deBoer:2005pt}. We will come back to these in more detail in a future publication \cite{delaOssa}, but we note that by a similar computation as that of \eqref{eq:FcheckLambda}, we can easily show that 
\begin{equation*}
{\cal\check A}_{\mathbf{14}}^1(Y,TY)\subseteq\ker(\check{\cal F})\:.
\end{equation*}
We can hence extend the notion of ${\cal TM}$ to include both the metric deformations and the $B$-field deformations, with the further requirement from the instanton condition that we need to restrict to elements $\Delta_t\in H^1_{\cd_\theta}(Y,TY))$ in the kernel of ${\cal \check F}$. Note that $\Delta_t$ can in principle include the $B$-field deformations in the $\mathbf{14}$-representation as well.

We can rephrase this result in terms of a cohomology group defined on an extension bundle $E$. Define the bundle $E$ which is the extension of $TY$ by the bundle ${\rm End}(V)$, given by the short exact sequence
\begin{equation}
\label{eq:SeS}
0 \longrightarrow {\rm End}(V)\longrightarrow E \longrightarrow TY\longrightarrow 0~,
\end{equation}
with extension class $\check{\cal F}$, and a connection ${\cal D}_E$ on the bundle $E$
\begin{equation*}
{\cal D}_E = \left(
\begin{array}{cc}
\check \dd_A &  \check{\cal F}\\
0 &  \cdth
\end{array}
\right)~.
\end{equation*}
It is not too difficult to show that this connection satisfies ${\cal D}_E^2= 0$ by equation \eqref{eq:idcohom}. The resemblance of the above sequence with that of the Atiyah algebroid \cite{atiyah1957complex} is clear, and it tempting to suggest that the infinitesimal moduli space is counted by the first cohomology as in that case. Let us see if this is correct.

Consider the cohomology group $H_{{\cal D}_E}^1(Y, E)$ and let
\[ x = \left(
\begin{array}{c}
\alpha\\
\Delta
\end{array}
\right)\in H_{{\cal D}_E}^1(Y,E)
~.
\]
where $\alpha$ is a one form with values in ${\rm End}(V)$ and $\Delta$ is a one form with values in $TY$.
Then a ${\cal D}_E$-closed one form $x$ is equivalent to
\begin{equation*}
{\cal D}_E\, x =
\left(
\begin{array}{cc}
\check \dd_A &  \check{\cal F}\\
0 &  \check \dth
\end{array}
\right)\,
\left(
\begin{array}{c}
\alpha\\
\Delta
\end{array}
\right)
=
\left(
\begin{array}{c}
\check \dd_A\alpha + \check{\cal F}(\Delta)\\
\cdth\Delta
\end{array}
\right)\,= 0
~.
\end{equation*}
That is
\[ \check \dd_A\alpha + \check{\cal F}(\Delta) = 0~,\qquad
\cdth\Delta = 0~.\]
These are just the equations which must be satisfied by the moduli of the instanton connection on the bundle $V$ over $Y$, together with the variations of the $B$-field and 
variations of the $G_2$ holonomy structure on $Y$ which preserve the instanton conditions.
Consider now one forms $x$ which are ${\cal D}_E$-exact
\[ x = {\cal D}_E \lambda~,\]
for some section $\lambda$ of $E$.
Let
\[ \lambda
= \left(
\begin{array}{c}
\epsilon\\
\delta
\end{array}
\right)~.\]
Then
\[ \check \dd_A\epsilon + \check{\cal F}(\delta) = \alpha~,\qquad
\cdth\delta = \Delta~.\]
Modulo such ${\cal D}_E$-exact terms, the second equation then tells us that 
\begin{equation*}
\Delta\in H^1_{\cd_\theta}(Y,TY))\:,
\end{equation*}
which are precisely the metric deformations preserving $G_2$-holonomy and $B$-field deformations inclusive as described above. Fixing the gauge of these deformations by e.g. considering harmonic forms, we are free to set $\cdth\delta=0$. However, as there are no globally covariantly constant vector fields on a manifold of $G_2$-holonomy, we have
\begin{equation*}
H^0_{\cd_\theta}(Y,TY) = 0\:.
\end{equation*}
It follows that $\delta=0$. The first equation then says that the bundle moduli should be modded out by the remaining gauge symmetries, which are given by exact forms $\cd_A\epsilon$. In summary, we can claim that the infinitesimal moduli space is given by
\begin{equation}
\label{eq:kernel2}
{\cal T}{\cal M} =H^1_{\check\dd_A}(Y, {\rm End}(V))\oplus {\rm ker}(\check{\cal F})=H^1_{{\cal D}_E}(Y, E)~.
\end{equation}
Indeed, this can be seen by computing the long exact sequence in cohomology associated to  the short exact sequence $E$
\begin{align*}
0\rightarrow H^1_{\cd_A}(Y,\End(V))\xrightarrow{i} H^1_{{\cal D}_E}(E)\xrightarrow{p}H^1_{\cd_\theta}(Y,TY))\xrightarrow{\check{\cal F}} H^2_{\cd_A}(Y,\End(V))\rightarrow...\:.
\end{align*}
We can the compute $H^1_{{\cal D}_E}(E)$ using exactness of the sequence. That is, we have
\begin{equation*}
H^1_{{\cal D}_E}(Y,E)\cong\textrm{Im}(i)\oplus\textrm{Im}(p)
\end{equation*}
Indeed by injectivity of the first map we see that $\textrm{Im}(i)\cong H^1_{\cd_A}(Y,\End(V))$, while $\textrm{Im}(p)=\ker(\check{\cal F})$. The result \eqref{eq:kernel2} follows.

It is interesting to see that the infinitesimal deformations of the extension bundle $E$, defined by the differential ${\cal D}_E$, are computed exactly as in the even-dimensional holomorphic case \cite{atiyah1957complex} by the first cohomology $H^1_{{\cal D}_E}(Y,E)$. Indeed, the bundle valued cohomologies we have defined in this paper have many similarities with their holomorphic cousins. We will study many of these similarities further in \cite{delaOssa}.

\subsection{Higher order obstructions and integrability}
\label{sec:higherdef}
Let us now go a step further and consider obstructions to higher order deformations of the instanton bundles. To do so, we will keep the $G_2$ geometry fixed for now. We will return to higher order deformations of the instanton condition together with the base, or equivalently the above defined extension $E$ in a future publication. In this section, we will also return to setups where the base geometry is some integrable $G_2$-structure manifold, generalising the $G_2$-holonomy condition.

We Let $\{A,B,C,..\}$ denote an infinitesimal direction in the vector space spanned by $H^1_{\cd_A}(Y,\End(V))$. As we saw in section \ref{sec:BundleCohomology}, the triple $(\check\Lambda^*(\End(V)),\cd_A,[\cdot,\cdot])$ forms a differentially graded Lie algebra. Furthermore, inserting a finite change of the connection
\begin{equation*}
A\rightarrow A+\Delta A
\end{equation*}
into the instanton condition produces the following condition on $\Delta A$
\begin{equation}
\label{eq:MA}
\cd_A\Delta A+\frac{1}{2}[\Delta A,\Delta A]=0\:.
\end{equation}
That is, $\Delta A$ should be a Maurer-Cartan element of the differentially graded Lie algebra, as is usual when one studies these kinds of deformation problem. Let $X^A$ correspond to the bundle moduli. We now assume that $\Delta A$ can be expanded in moduli fields as
\begin{equation*}
\Delta A=X^A\p_AA+\frac{1}{2}X^AX^B\,\p_A\p_BA+...\:.
\end{equation*}
Since the $X^A$ are arbitrary, if we plug this expansion back into \eqref{eq:MA} we must have
\begin{align*}
\cd_A\p_B A&=0\\
\cd_A\p_B\p_CA+[\p_BA,\p_CA]&=0\:,
\end{align*}
and so on. The first equation is just the statement that the infinitesimal deformations take values in $H^1_{\cd_A}(Y,\End(V))$. The second equation gives the first obstruction to these deformation. Indeed, recall from Theorem \ref{tm:ringA}, that $[\cdot,\cdot]$ is a  well-defined product in cohomology. Thus, this product of infinitesimal variations of $A$ is required to vanish in cohomology, otherwise we have an obstruction to the infinitesimal deformations of $A$ at second order in perturbation theory. Note that higher order deformations give higher order obstructions in a similar fashion. In string compactifications, is expected that these obstructions correspond to Yukawa couplings in the lower-dimensional effective theory, a question we hope to return to in future publications.

Let us also take a moment to speculate about the behaviour of the higher order deformations when we also include deformations of the base. In this case, it is perhaps instructive to restrict to bundles over a base with $G_2$-holonomy, whose moduli space is unobstructed \cite{joyce1996:1,joyce1996:2,joyce2000}. We expect the full deformation problem to give rise to a similar Maurer-Cartan equation, but now with a $\Delta\in\Lambda^1(E)$, so that
\begin{equation}
{{\cal D}_E}\Delta+\frac{1}{2}[\Delta,\Delta]=0\:,
\end{equation}
where $[\cdot,\cdot]$ is an appropriate bracket on $\Omega^*(E)$ to be discerned. Should this happen, one would get that obstruction classes counted by $H^2_{{\cal D}_E}(Y,E)$, just like for the Atiyah algebroids in complex geometry. We hope to return to this question in the future.

%\newpage

\section{Conclusions and Outlook}
In this paper, we have studied the infinitesimal deformations of a pair $(Y,V)$, where $Y$ is a manifold of $G_2$ holonomy, and $V$ is a vector bundle with a connection satisfying the $G_2$ instanton condition. We found that the structure of the infinitesimal moduli space very much resembles that of the Atiyah algebroid in the case of holomorphic bundles over complex manifolds \cite{atiyah1957complex}. Indeed, we found that the infinitesimal geometric deformations of the base, corresponding to elements of $H^1_{\cd_\theta}(Y,TY)$, must be in the kernel of an appropriate $G_2$ generalisation of the Atiyah map, 
\begin{equation*}
\check{\cal F}\::\;\;\;H^1_{\cd_\theta}(Y,TY)\rightarrow H^2_{\cd_A}(Y,\End(V))\:,
\end{equation*}
just like in the holomorphic case. The map $\check{\cal F}$ is given in terms of the curvature of the bundle. 

This structure is very interesting and prompts further investigation. In particular, just as the K\"ahler condition on the base is not necessary in the holomorphic case, the $G_2$ holonomy condition can also be relaxed in the seven-dimensional case. Indeed, as we have seen we only need the base to have an {\it integrable $G_2$ structure} in order to define the $\check \dd_A$-cohomologies which are used in computing the infinitesimal deformations. We have taken some steps in this direction in the current paper, and will investigate this further in an upcoming publication \cite{delaOssa}. Furthermore, in order to make more contact with physics and the heterotic string, one also needs to consider the heterotic Bianchi Identity. We will investigate this further in \cite{delaOssa}, but give a brief prelude here. There is evidence that the combined structure fits neatly into a double extension of the form
\begin{equation*}
0\rightarrow T^*Y\rightarrow{\cal Q}\rightarrow E\rightarrow 0\:,
\end{equation*}
where $E$ is the $G_2$ Atiyah algebroid of section \ref{sec:bundmod}, just as in the holomorpic case of \cite{Anderson:2014xha, delaOssa:2014cia, Garcia-Fernandez:2015hja}. The corresponding extension map is defined by the Bianchi Identity. Equivalently, as in the holomorphic case we hope to show that the system of heterotic BPS-equations together with the Bianchi Identity can be used to construct a differential $\check{\cal D}$ on ${\cal Q}$, and that the infinitesimal moduli are counted by $H^1_{\check{\cal D}}({\cal Q})$ with respect to this differential. 

Having discussed the infinitesimal moduli space, we hope to also address the issue of higher order and integrable deformations. There is a lot of mathematical literature on the deformations of the holomorphic Atiyah algebroid, see e.g. \cite{atiyah1957complex, donaldson1989connected, huybrechts1995tangent, Anderson:2011ty}. Since the structure of the corresponding differential complexes and extensions are so similar in the $G_2$ case, there is hope that many of the results of the Atiyah algebroid can be carried over to the $G_2$ setting without too much effort. We hope to investigate some of these aspects in the future.

One other interesting application of the results of this paper comes when we consider reductions to $SU(3)$ structure three-folds. Indeed, upon reducing on
\begin{equation*}
Y=X_6\times\mathbb{R}\:,
\end{equation*}
where $X_6$ is a complex three-fold with an appropriate $SU(3)$ structure, it is easy to see that the seven-dimensional instanton condition splits into the requirement that the bundle is holomorphic, in addition to the Yang-Mills condition,
\begin{equation*}
F^{(0,2)}=0\:,\;\;\;g^{a\bar b}F_{a\bar b}=0\:.
\end{equation*}
Ignoring issues related to compactness of $Y$, it is interesting to see how both these conditions can be incorporated in the same structure $\check\dd_A$. Due to this fact, it is also conceivable that one can learn a lot about the $\check\dd_A$-cohomologies by what is already known about stable holomorphic bundles, and this is an interesting direction of further investigation.

\section*{Acknowledgements}
The authors would like to thank Mario Garcia-Fernandez, Spiro Karigiannis, George Papadopoulos, Carlos Shahbazi and Carl Tipler  for interesting discussions. ML and EES thank the organisers for the 2016 Benasque workshop {\it Superstring solutions, supersymmetry and geometry} where part of this work was finalised. ML acknowledges partial support from the COST Short Term Scientific Mission MP1210-32976, and thanks the Mathematical Institute at Oxford University for hospitality  during part of the preparation of this paper. XD and EES would like to thank the Department of Physics and Astronomy at Uppsala University for its hospitality when part of the work was finalised. The work of EES, made within the Labex Ilp (reference Anr-10-Labx-63), was supported by French state funds managed by the Agence nationale de la recherche, as part of the programme Investissements d'avenir under the reference Anr-11-Idex-0004-02.  The work of XD is supported in part by the EPSRC grant BKRWDM00.

.

% ---------------------------------------- Appendix

\begin{appendix}
%\newpage

\section{Formulas}\label{app:formulas}

In this appendix we gather a number of formulas and identities which we have used in this paper.
The contractions between $\varphi$ and $\psi$ can be found in a number of references, see for example \cite{Bryant:2005mz}. More useful relations can be found in \cite{delaOssa:2014lma} 
\begin{align}
\varphi^{abc}\, \varphi_{abc} &=42~,
\\
\varphi^{acd}\, \varphi_{bcd} &= 6\, \delta_b^a~,
\label{eq:g2ident2}\\
\varphi^{eab}\, \varphi_{ecd} &= 2\, \delta_{[c}^a\, \delta_{d]}^b
+ \psi^{ab}{}_{cd}~.
\label{eq:g2ident3}
\\[7pt]
\varphi^a{}^{d_1 d_2}\, \psi_{bc d_1 d_2} &= 4\, \varphi^a{}_{bc}~,
\label{eq:g2ident4}\\
\varphi^{abf}\, \psi_{cdef} &= - 6\, \delta^{[a}{}_{[c}\, \varphi^{b]}{}_{de]}~,
\label{eq:g2ident5}
\\[7pt]
\psi^{abcd}\psi_{abcd} &= 7\cdot 24 = 168~,
\\
\psi^{acde}\psi_{bcde} &= 24\, \delta_b^a~,\\
\psi^{abe_1e_2}\psi_{cde_1e_2} &= 8\, \delta_{[c}^a\, \delta_{d]}^b + 2\,\psi^{ab}{}_{cd}~,\\
\psi^{a_1a_2a_3c}\psi_{b_1b_2b_3c} &= 6\, \delta_{[b_1}^{a_1}\, \delta_{b_2}^{a_2} \, \delta_{b_3]}^{a_3}
+ 9\, \psi^{[a_1a_2}{}_{[b_1b_2}\, \delta^{a_3]}_{b_3]} - \varphi^{a_1a_2a_3}\,\varphi_{b_1b_2b_3}~,
\\
\psi^{a_1a_2a_3a_4}\psi_{b_1b_2b_3b_4} &= 24\, \delta_{[b_1}^{a_1}\, 
\delta_{b_2}^{a_2} \, \delta_{b_3}^{a_3}\, \delta_{b_4]}^{a_4}
\\[2pt]
&\qquad
+ 72\, \psi^{[a_1a_2}{}_{[b_1b_2}\, \delta_{b_3}^{a_3}\,\delta^{a_4]}_{b_4]}- 16\,  \varphi^{[a_1a_2a_3}\,\varphi_{[b_1b_2b_3}\, \delta^{a_4]}_{b_4]}~.
\end{align}

Let $\alpha$ be a one form (possibly with values in some bundle)
\begin{align}
\varphi\lrcorner(\alpha\wedge\varphi) &= (\alpha\lrcorner\psi)\lrcorner\psi = -4\, \alpha~, 
\label{eq:onepsipsi}\\
\psi\lrcorner(\alpha\wedge\psi) &= (\alpha\lrcorner\varphi)\lrcorner\varphi =3\, \alpha~, 
\label{eq:onephiphi}\\
\varphi\lrcorner(\alpha\wedge\psi) &= 2\, \alpha\lrcorner\varphi~.
\label{eq:onephipsi}
\end{align}

Let $\alpha$ be a two form (possibly with values in some bundle)
\begin{align}
\varphi\lrcorner(\alpha\wedge\varphi) &= -\, (\alpha\lrcorner\psi)\lrcorner\psi = 2\, \alpha + \alpha\lrcorner\psi~,
\label{eq:twopsipsi}\\
\psi\lrcorner(\alpha\wedge\psi) &= (\alpha\lrcorner\varphi)\lrcorner\varphi = 3\, \pi_7(\alpha) = 
\alpha + \alpha\lrcorner\psi~.\label{eq:twophiphi}
\end{align}

\begin{lemma}
Let $\beta\in \Lambda^2_{14}$  (possibly with values in some bundle).  Then the identities
\begin{align}
 \label{eq:beta14id}
\beta_{d[a}\, \varphi_{bc]}{}^d &= 0~,\\
 \beta^{ae}\, \psi_{ebcd} &= - 3\, \beta_{e[b}\, \psi_{cd]}{}^{ae}~,
 \label{eq:betapsi}
 \end{align}
follow from $\beta = - \beta\lrcorner\psi$.  
\end{lemma}
\begin{proof}
\[
\beta = - \beta\lrcorner\psi\quad\Longrightarrow\quad -2\, \beta_{ab} = \beta_{cd}\, \psi^{cd}{}_{ab}~.
\]
We use this to prove both identities.

For the first one, we contract with $\varphi$ as follows:
\begin{align*}
\beta_{ad}\, \varphi^d{}_{bc} &= - \frac{1}{2}\, \beta_{ef}\, \psi^{ef}{}_{ad}\,\varphi^d{}_{bc}
= 3\, \beta_{ef}\, \delta_{[b}{}^{[e}\, \varphi_{c]}{}^{fd]}\, g_{da}
= 2\, g_{da}\, \beta_{ef}\, \delta_{[b}{}^e\, \varphi_{c]}{}^{fd}\\
&= - \beta_{bd}\, \varphi^d{}_{ca} - \beta_{cd}\varphi^d{}_{ab}~.
\end{align*}

For the second identity, we contract with $\psi$ as follows:
\begin{align*}
-2\, \beta^{ae}\, \psi_{ebcd} &= \beta_{f_1f_2}\, \psi^{f_1f_2ae}\, \psi_{ebcd}
\\
& = - \beta_{f_1f_2}\, (6\, \delta^{[f_1}_b\, \delta^{f_2}_c\, \delta^{a]}_d
+ 9\, \psi^{[f_1f_2}{}_{[bc}\, \delta^{a]}_{d]})
\\
&= - 6 \, \beta_{[bc}\, \delta^a_{d]}
 - 3\, \beta_{f_1f_2}\, ( \psi^{f_1f_2}{}_{[bc}\, \delta^{a}_{d]} + 2 \psi^{f_2a}{}_{[bc}\, \delta^{f_1}_{d]})
 \\
&= - 6 \, \beta_{[bc}\, \delta^a_{d]} + 6 \, \beta_{[bc}\, \delta^a_{d]}
- 6\, \psi^{ea}{}_{[bc}\, \beta_{d]e}~.
\end{align*}
Hence
\[ \beta^{ae}\, \psi_{ebcd} = - 3\, \psi^{ea}{}_{[bc}\, \beta_{d]e}
= - 3\, \beta_{e[b}\, \psi_{cd]}{}^{ae}
~.\]

\end{proof}

\begin{lemma}\label{lem:idfortwo14}
Let $\beta\in\Lambda^2_{14}(Y,E)$, where $E$ is a bundle over $Y$, so that
\[ \beta\wedge\psi = 0~.\]
Then, this equation is equivalent to 
\begin{equation}
\frac{1}{3!}\,\beta\wedge\psi_{bcda}\, \dd x^{bcd} 
= \beta_{ab}\,\dd x^b\wedge\psi~.\label{eq:idtwo14}
\end{equation}
Hence
\[ \pi_{14}(\beta\wedge\psi_{bcda}\, \dd x^{bcd}) = 0~.\]
\end{lemma}
\begin{proof}
Equation $\beta\wedge\psi = 0$ can be written as
\[0= 6\, \beta_{[e_1e_2}\wedge\psi_{bcda]}~.\]
Then
\begin{align*}
0&= 3\, \beta_{[e_1e_2}\wedge\psi_{bcda]}\, \dd x^{e_1e_2bcd} =
(3\, \beta_{e_1e_2}\, \psi_{bcda} - \beta_{ae_1}\, \psi_{e_2bcd})\, \dd x^{e_1e_2bcd}
\\
&= 4\, \beta\wedge \psi_{bcda}\, \dd x^{bcd} - 4!\, \beta_{ab}\, \dd x^b\wedge\psi~,
\end{align*}
hence equation \eqref{eq:idtwo14} follows.  The  statement that
\[  \pi_{14}(\beta\wedge\psi_{bcda}\, \dd x^{bcd}) = 0~,\]
is just the observation that the right hand side of
equation \eqref{eq:idtwo14} is a five-form in $\Lambda^5_{7}$.

\end{proof}

\begin{lemma}
Let $\beta\in\Lambda^2_7$ (perhaps $\beta$ also takes values in some vector bundle).  Then $\beta$ satisfies the identity
\begin{equation}
\beta_{da}\, \varphi_{bc}{}^d = 
\beta_{d[b}\, \varphi_{c]a}{}^d  + g_{\varphi\, a[b}\, (\beta\lrcorner\varphi)_{c]}~.
\label{eq:id2rep7}
\end{equation}
\end{lemma}

\begin{proof}
A two form $\beta$ in $\Lambda^2_7$, must satisfy
\[ 2\, \beta = \beta\lrcorner\psi~,\]
or, equivalently
\[ 4\, \beta_{ad} = \beta_{bc}\, \psi^{bc}{}_{ad}~.\]
Contracting this with $\varphi_{bc}{}^d$ and using equation \eqref{eq:g2ident5} we have
\begin{align*}
4\, \beta_{ad}\, \varphi_{bc}{}^d &=
\beta_{e_1e_2}\, \psi^{e_1e_2}{}_{ad}\, \varphi_{bc}{}^d
= -6\, \beta_{e_1e_2}\, g_{\varphi\, a e_3}\,
\delta_{[b}^{[e_1}\, \varphi_{c]}{}^{e_2 e_3]} 
\\ 
&= -2 \, \beta_{e_1e_2}\,g_{\varphi\, a e_3}\, 
( 2\, \delta_{[b}^{e_1}\, \varphi_{c]}{}^{e_2 e_3} 
+ \delta_{[b}^{e_3}\, \varphi_{c]}{}^{e_1 e_2})
\\
&= -4\, \beta_{d [b}\,\varphi_{c]a}{}^d 
-2\, \beta_{de}\, g_{\varphi\, a[b}\, \varphi_{c]}{}^{de}
= -4 ( \beta_{d [b}\,\varphi_{c]a}{}^d 
+ g_{\varphi\, a[b}\, (\beta\lrcorner\varphi)_{c]})~,
\end{align*}
from which the identity \eqref{eq:id2rep7} follows. 
\end{proof}

\subsection{Identities involving Hodge duals}
Let $\alpha$ be a $k$-form and $\beta$ a $p+k$-form. Then
\be 
\alpha\lrcorner \beta = (-1)^{p(d-p-k)}\, *(\alpha\wedge*\beta)~,
\ee
\be
\label{eq:ddaggeralpha}
\dd^\dagger\alpha = - \frac{1}{(k-1)!}\,  g^{mn}\,\nabla^{LC}_m ( \alpha_{np_1\cdots p_{k-1}})
\, \dd x^{p_1}\wedge\cdots\wedge\dd x^{p_{k-1}}~.\ee

\subsection{Identities for derivatives of $\varphi$ and $\psi$}
 In our computations related to the moduli problem, we will need identities which relate derivatives of $\varphi$ and $\psi$ with the connection $\dth$.  We present these here by means of two lemmas. Note that the lemmas hold for manifolds with any $G_2$ structure, not only for $G_2$ holonomy. 
  
 \begin{lemma}\label{lem:idone}
Let $\psi$ be a four form on a seven dimensional manifold $Y$ (not necessarily a manifold with a $G_2$ structure).  Then we have the identity
\begin{equation*}
4\,\dd(\psi_{bcda}\, \dd x^{bcd})
 =  - 4!\, \partial_a\psi  + (\dd\psi)_{bcdea}\, \dd x^{bcde}~.
\end{equation*}
If the manifold $Y$ has a $G_2$ structure determined by $\varphi$ (not necessarily harmonic) with a connection $\nabla$ compatible with the $G_2$
structure we also have
\[ \partial_a\psi = - \frac{1}{3!}\, \theta_a{}^b\wedge \psi_{cdeb}\,\dd x^{cde}~.\]
where $\theta_a{}^b = \Gamma_{ac}{}^b\, \dd x^c$ and
$\Gamma_{ac}{}^b$ are the connection symbols of $\nabla$.
\end{lemma} 
\begin{proof}
For the first identity we compute
\begin{align*}
\dd(\psi_{bcda}\, \dd x^{bcd}) &= \partial_e\psi_{bcda}\, \dd x^{ebcd}
\\
&= (5\, \partial_{[e}\psi_{bcda]} - 3\, \partial_b\psi_{cdae} - \partial_a\psi_{ebcd})\, \dd x^{ebcd}
\\
&= (\dd\psi)_{bcdea}\, \dd x^{bcde} - 3\, \dd(\psi_{bcda}\, \dd x^{bcd})
- 4!\, \partial_a\psi~,
\end{align*}
which gives the result desired.

For the second we have, 
\begin{align*}
 \partial_a\psi &= \frac{1}{4!}\, \partial_a\psi_{bcde}\,\dd x^{bcde}
 = \frac{1}{4!}\, (\nabla_a\psi_{bcde} + 4\, \Gamma_{ab}{}^f\, \psi_{fcde})\,\dd x^{bcde}
 \\[3pt]
 &= 
 \frac{1}{3!}\, \Gamma_{ab}{}^f\, \psi_{fcde}\,\dd x^{bcde}
 = - \frac{1}{3!}\, \theta_a{}^b\wedge \psi_{cdeb}\,\dd x^{cde}~,
 \end{align*}
 where we have used $\nabla\psi = 0$. 
 
\end{proof}

\begin{lemma}\label{lem:idtwo}
Let $\varphi$ be a three form on a seven dimensional manifold $Y$ (not necessarily a manifold with a $G_2$ structure).  Then we have the identity
\begin{equation*}
3\,\dd(\varphi_{bca}\, \dd x^{bc})
 =  3!\, \partial_a\varphi  + (\dd\varphi)_{bcda}\, \dd x^{bcd}~.
\end{equation*}
If the manifold $Y$ has a $G_2$ structure determined by $\varphi$ (not necessarily harmonic) with a connection $\nabla$ compatible with the $G_2$
structure we also have
\[ \partial_a\varphi =  \frac{1}{2}\, \theta_a{}^b\wedge \varphi_{cdb}\,\dd x^{cd}~.\]
where $\theta_a{}^b = \Gamma_{ac}{}^b\, \dd x^c$ and
$\Gamma_{ac}{}^b$ are the connection symbols of $\nabla$.
\end{lemma} 

\begin{proof}
The proof is similar to the proof of lemma \ref{lem:idone} and is left as an exercise.
\end{proof}

%\newpage

\section{Elliptic Complex}
\label{app:Elliptic}

In this appendix we recall basic notions about ellipticity of operators and complexes, following the book of Gilkey \cite{gilkey_online}, to which we refer for a more detailed account. We then show that the complexes defined in section \ref{sec:g2struc} and \ref{sec:instbund} are elliptic. We first recall that a complex is elliptic if it is exact on the level of leading symbols:

\begin{definition}[Gilkey]
Let $V$ be a graded vector bundle: $V$ is a collection of vector bundles $\{V_j\}_{j\in \mathbb{Z}}$ such that $V_j \neq \{0\}$ for only a finite number of indices $j$. Let $P$ be a graded pseudo differential operator ($\Psi$DO) of order $d$: $P$ is a collection of $d^{th}$ order $\Psi$DOs $P_j: C^{\infty}(V_j) \to C^{\infty}(V_{j+1})$. Then $(P,V)$ is a {\it complex} if $P_{j+1} P_j = 0$ and $\sigma_L P_{j+1} \sigma_L P_j = 0$. $(P,V)$ is an {\it elliptic complex} if 
\be \label{eq:ellip}
N(\sigma_L P_j)(x,\xi) = R (\sigma_L P_{j+1}) (x,\xi) \ee
i.e. the complex is exact on the level of the leading symbol $\sigma_L$.\footnote{$N$ denotes the null space (i.e. kernel) of the operator, and $R$ denotes the range (or image).} 
\end{definition}

To be able to use this definition, we must define pseudo differential operators, as well as the symbol of an operator. Let  $\alpha = (\alpha_1, ... \alpha_m)$ a multi-index, and $|\alpha|= \alpha_1 + ...+\alpha_m$. Introduce a notation for multiple partial differential operator
\[D^{\alpha}_x = (-i)^{|\alpha|} (\frac{\partial}{\partial x_1})^{\alpha_1} ... (\frac{\partial}{\partial x_m})^{\alpha_m} \; ,
\]
where the factors of $i$ will simplify the expressions below. A  linear partial differential operator of order $d$ may then be expressed as
\[
P =  \sum_{|\alpha| \le d} a_{\alpha}(x) D^{\alpha}_x \; ,
\]
where $a_{\alpha}(x)$ are smooth functions. Using the dual variable $\xi$ that appears in the Fourier transform (where the measure $dx$ is defined so that factors of $2 \pi i$ are absorbed)
\be \label{eq:fourier}
\hat{f}(\xi) = \int e^{-i x \cdot \xi} f(x) d x \; ,
\ee
we define the symbol $\sigma P$  by
\be
\sigma P (x, \xi) = \sum_{|\alpha| \le d} a_{\alpha}(x) \xi^{\alpha}
\ee
This is a polynomial of order $d$ in $\xi$. The leading symbol, $\sigma_L P$, is the highest degree part of this polynomial
\be
\sigma_L P(x, \xi) = \sum_{|\alpha| = d} a_{\alpha}(x) \xi^{\alpha} \;.
\ee
As described in \cite{gilkey_online}, we may generalise to  non-polynomial symbols, for which the corresponding operator is called a pseudo differential operator ($\Psi$DO). The complexes of relevance here all have polynomial symbols.

The condition of ellipticity for a complex can also be stated in terms of a constraint on the associated Laplacian. Define the Laplacian of an operator $P_j$ as 
\be
\Delta_j = P_j^* P_j + P_{j-1} P_{j-1}^*: C^{\infty}(V_j) \to C^{\infty}(V_j) \; ,
\ee
where $P_j^*$ denotes the adjoint operator of $P_j$. The leading symbol of $\Delta_j$ is given by
\be
\sigma_L(\Delta_j) = p_j^* p_j + p_{j-1}p_{j-1}^*, \mbox{ where } p_j = \sigma_L(P_j)
\ee
An operator is elliptic if its leading symbol is non-singular for $\xi \neq 0$. We then have, as proven in \cite{gilkey_online}:
\begin{lemma}[Gilkey]
\label{lem:laplace} 
Let $(P,V)$ be a $d^{th}$ order partial differential complex. Then $(P,V)$ is elliptic iff $\Delta_j$ is an elliptic operator of order $2d$ for all $j$.
\end{lemma}

\subsection{Examples of elliptic complexes}
Let us now demonstrate that the complexes \eqref{eq:dolb} and \eqref{eq:dolbV} are elliptic. We will use that de Rham complex is elliptic, so we start by recalling this fact.

\paragraph{Ellipticity of the de Rham complex}(Gilkey)

\noindent
To prove that the de Rham comples is elliptic, we need the symbol of $\dd$. 
Recall that the symbol of $\frac{\partial}{\partial x_j}$ is given by the dual coordinate $i \xi_j$, via the Fourier transform \eqref{eq:fourier}.
To find the symbol for $\dd = \sum \dd x_j \wedge \frac{\partial}{\partial x_j}$, define the one-form $\xi = \sum \xi_j \dd x_j$. The symbol of $\dd$ is then
\be
\sigma(\dd) = i\rm{ext}(\xi) 
\ee 
where
\be
\rm{ext}(\xi) \omega = \xi \wedge \omega \; .
\ee
Note that for this example, the symbol only contains monomials of maximal degree = 1, and so the leading symbol coincides with the symbol.

We must now show that \eqref{eq:ellip} holds for $\sigma(\dd) = i\rm{ext}(\xi)$. Fix $\xi\neq0$ and choose a frame $\{e_1, ..., e_m\}$ in $T^*Y$ such that $\xi = e_1$. We then have
\be
i{\rm ext}(\xi) e_I = \Big{\{}
\begin{matrix}
0 &, i_1 = 1\\
e_J &, J = \{1,i_1,...,i_p\}
\end{matrix} \; .
\ee
It follows that $N\left({\rm ext}(\xi_p)\right) = R\left({\rm ext}(\xi_{p-1})\right)$, as required (note that the index $p$ and $p-1$ is not necessary, since the operator $\dd$ is the same for all $p$). This proves that the de Rham complex is exact on the symbol level, as required for an elliptic complex.

\paragraph{Ellipticity of the $\check{\dd}$-complex}
Let us now consider the complex \eqref{eq:dolb} of Fernandez-Ugarte. The ellipticity of this complex was shown in \cite{reyes1993some, carrion1998generalization}, but we repeat the argument here for convenience. We first recall the complex
\begin{equation}
\label{eq:dolb2}
0\rightarrow\Lambda^0(Y)\xrightarrow{\cd}\Lambda^1(Y)\xrightarrow{\cd}\Lambda^2_7(Y)\xrightarrow{\cd}\Lambda^3_1(Y)\rightarrow0
\end{equation}
where the corresponding nilpotent operator is
\begin{equation}
\cd=\pi\circ\dd\:,
\end{equation}
where $\pi$-denotes the projection onto the appropriate group. Note that for $\alpha\in\check\Lambda^*(Y)$ we have
\begin{equation}
\psi\wedge\alpha=\psi\wedge\pi(\alpha)\:.
\end{equation}
The corresponding symbol is hence
\begin{equation}
\sigma(\xi,x)=\pi\circ\rm{ext}({\xi})\:,
\end{equation}
modulo non-important prefactors.

Now let $\alpha$ be a one-form. If $\alpha$ is in the kernel of $\sigma$, it implies that
\begin{equation}
\psi\wedge\xi\wedge\alpha=0\:,
\end{equation}
or equivalently
\begin{equation}
\xi^m\alpha^n\varphi_{mnp}=0\:.
\end{equation}
It follows that $\alpha=f\,\xi$ for some function $f$, and $N(\sigma)=R(\sigma)$ at the level of one-forms. We next assume $\alpha\in\Lambda^2_7(Y)$. Being in the kernel of $\sigma$ is then equivalent to
\begin{equation}
\xi^m\alpha^{np}\varphi_{mnp}=0\:.
\end{equation}
Further, we can decompose $\alpha$ as
\begin{equation}
\alpha=\xi\wedge\gamma+\beta\:,
\end{equation}
where $\xi\lrcorner\beta=0$. It follows that
\begin{equation}
\xi^m\beta^{np}\varphi_{mnp}=0\:.
\end{equation}
As $\beta$ is orthogonal to $\xi$, this implies that
\begin{equation}
\beta^{np}\varphi_{mnp}=0\:,\;\;\;\Rightarrow\;\;\;\beta\wedge\psi=0\:.
\end{equation}
But then
\begin{equation}
\alpha=\pi(\alpha)=\pi(\xi\wedge\gamma)\:.
\end{equation}
Hence $\alpha\in\rm{Im}(\sigma)$. Finally, we note that $\sigma$ is surjective onto $\Lambda^3_1(Y)$, and so the symbol is exact at this level as well. Hence the complex \eqref{eq:dolb2} is elliptic. 

\paragraph{Ellipticity of the $\cd_A$-complex}
 It is now straight-forward to prove also that the complex \eqref{eq:dolbV} is elliptic. First, recall that  by Theorem \ref{th:dacheck}, we have
\be
\cd_A^2 = 0
\ee
as long as $A$ is an instanton connection. With
\be
\dd_A \omega = \sum_{j,I} \left(\frac{\partial \omega_I}{\partial x_j} + A_j \omega_I \right) \dd x_j \wedge \dd x^I \; ,
\ee
the symbol of $\cd_A$ is 
\be
\sigma(\cd_A) = \pi \circ {\rm ext}(i \xi+A)
\ee
so the leading symbol $\sigma_L(\cd_A)$ equals $\sigma_L(\cd)$, and hence also $(\cd_A, C^{\infty}(\Lambda(T^*M)))$ is an elliptic complex.

\end{appendix}

% ---------------------------------------- References
\newpage

%\bibliographystyle{JHEP}

%\bibliography{bibliography}

\providecommand{\href}[2]{#2}\begingroup\raggedright\endgroup

% ---------------------------------------- References
\end{document}